\documentclass{article}

 \usepackage[preprint]{neurips_2025}

\usepackage[utf8]{inputenc} 
\usepackage[T1]{fontenc}    
\usepackage{hyperref}       
\usepackage{url}            
\usepackage{booktabs}       
\usepackage{amsfonts}       
\usepackage{nicefrac}       
\usepackage{microtype}      
\usepackage{xcolor}         
\usepackage{enumitem}
\usepackage{amsthm}
\usepackage{amsmath}
\usepackage{graphicx}
\usepackage{subcaption}
\usepackage{multirow} 
\usepackage{adjustbox} 
\newtheorem{theorem}{Theorem}[section]        

\newtheorem{proposition}[theorem]{Proposition}

\title{Logic-Gated Time-Shared Feedforward Networks for Alternating Finite Automata: Exact Simulation and Learnability}

\author{%
  Sahil Rajesh Dhayalkar \\
  Arizona State University\\
  \texttt{sdhayalk@asu.edu} \\
}

\begin{document}
\theoremstyle{definition}
\newtheorem{definition}[theorem]{Definition}
\newtheorem{remark}[theorem]{Remark}

\maketitle


\begin{abstract}
We present a formal and constructive framework for simulating Alternating Finite Automata (AFAs) using Logic-Gated Time-Shared Feedforward Networks (LG-TS-FFNs). Unlike prior neural automata models limited to Nondeterministic Finite Automata (NFAs) and existential reachability, our architecture integrates learnable, state-dependent biases that function as differentiable logic gates, enabling the representation of both Existential \textsc{\textsc{OR}} and Universal \textsc{\textsc{AND}} aggregation within a shared-parameter linear recurrence. We prove that this architectural modification upgrades the network's computational class to be structurally isomorphic to AFAs, thereby inheriting their exponential succinctness: the network can represent regular languages requiring $2^n$ states in an NFA with only $n$ neurons. We rigorously establish that the forward pass of an LG-TS-FFN exactly simulates the reachability dynamics of an AFA, including instantaneous $\varepsilon$-closures. Furthermore, we demonstrate empirical learnability: a continuous relaxation of the logic gates allows the network to simultaneously recover the automaton's topology and logical semantics from binary labels via standard gradient descent. Extensive experiments confirm that our model achieves perfect recovery of ground-truth automata, bridging the gap between statistical learning and succinct, universal logical reasoning.
\end{abstract}

\textbf{Keywords:} Alternating Finite Automata simulation, Automata theory, feedforward neural networks, Gradient-Based Learnability of Automata, interpretable models, learnability, Logic-Gated Time-Shared Depth-Unrolled Feedforward Networks, matrix-vector products, symbolic computation.

\section{Introduction}
\label{sec:introduction}
The remarkable empirical success of deep learning has been driven largely by the scalability of architectures such as Transformers and Recurrent Neural Networks (RNNs) on massive datasets. Despite these advances, a fundamental tension remains between the statistical nature of neural networks and the discrete, symbolic nature of formal languages. While neural networks excel at fuzzy pattern matching and generalization, they often struggle with tasks requiring strict logical adherence, such as unbounded counting, hierarchical recursion, or exact boolean reasoning. This dichotomy has fueled a resurgence of interest in understanding the theoretical expressivity of neural models through the lens of Automata Theory.

Historically, this inquiry has focused on two primary classes of automata: Deterministic Finite Automata (DFAs) and Nondeterministic Finite Automata (NFAs). Early work by \cite{OMLIN199641} established that RNNs can simulate DFAs, but extracting these automata from trained networks often leads to an explosion in state space complexity \cite{hopcroft2001introduction}. More recently, the Time-Shared Feedforward Networks (TS-FFNs) demonstrated that it can efficiently simulate NFAs, DFAs and Probabilistic Finite Automata by exploiting the high-dimensional superposition of state vectors \cite{dhayalkar2025nfa, dhayalkar2025neuralnetworksuniversalfinitestate,dhayalkar2025symbolicfeedforwardnetworksprobabilistic}. This finite automata-neural correspondence allows networks to represent existential logic ($\exists$), which is the ability to verify if at least one valid path exists, with exponential succinctness compared to DFAs.

However, a critical gap remains. Classical automata theory tells us that succinctness does not stop at nondeterminism. Alternating Finite Automata (AFAs)~\cite{Chandra1981} generalize NFAs by introducing Universal ($\forall$) branching, allowing states to enforce that all outgoing paths must lead to acceptance. This addition of conjunctive logic makes AFAs exponentially more succinct than NFAs and doubly exponentially more succinct than DFAs \cite{Fellah01011990}, providing a powerful formalism for modeling complex dependencies such as synchronization, safety constraints, and logical entailment.

In this paper, we bridge this gap by introducing the Logic-Gated Time-Shared Feedforward Network (LG-TS-FFN). We prove that by equipping standard linear layers with a learnable, state-dependent bias term, we can transform neurons into differentiable logic gates capable of switching between existential (\textsc{OR}) and universal (\textsc{AND}) aggregation. This seemingly simple modification has profound theoretical implications: it upgrades the computational class of the network from NFA-equivalent to AFA-equivalent, enabling it to represent regular languages with optimal succinctness.

Our contributions are threefold:
\begin{enumerate}
    \item Theoretical Equivalence: We prove that Logic-Gated TS-FFNs are structurally isomorphic to Alternating Finite Automata (Theorem \ref{THM:EQUIVALENCE}). We show that the forward pass of the network ($t=0 \to L$) exactly simulates the AFA's reachability dynamics, including the instantaneous propagation of logic via $\varepsilon$-closures (Theorem \ref{THM:AFA_SIMULATION}).
    \item Exponential Succinctness: We demonstrate that our architecture inherits the succinctness properties of AFAs. By learning to represent boolean intersections directly within the linear recurrence, the network can model patterns requiring $2^n$ states in an NFA using only $n$ neurons (Proposition \ref{PROP:SUCCINCTNESS}).
    \item Differentiable Learnability: We propose a continuous relaxation of the logic gates that allows the network to be trained end-to-end via gradient descent. We show that the network can simultaneously learn the topology of the automaton and the logical semantics of its states (\textsc{AND} vs. \textsc{OR}) from binary labels, achieving perfect recovery of complex ground-truth automata in controlled experiments (Proposition \ref{PROP:LEARNABILITY}).
\end{enumerate}

By establishing this exact correspondence, we provide a rigorous theoretical grounding for neuro-symbolic reasoning within standard deep learning primitives. Our results suggest that the dense connections of a neural layer are not merely feature extractors, but can function as a dynamic, differentiable boolean circuit, offering a path toward models that combine the learnability of neural networks with the succinctness and interpretability of formal logic.

\section{Related Work}
\label{sec:related_work}

The synthesis of neural networks and automata theory constitutes a foundational pillar of neuro-symbolic computing, tracing its lineage to the seminal demonstration by~\cite{McCulloch1943} that threshold networks are computationally equivalent to finite state machines. Our work advances this tradition by establishing a rigorous structural equivalence between linear-threshold networks and Alternating Finite Automata (AFAs), addressing critical questions of succinctness, logical expressivity, and learnability that remain open in the current literature.

\paragraph{Neural Representation of Formal Languages}
A substantial body of theoretical work investigates the computational power of Recurrent Neural Networks (RNNs). Siegelmann and Sontag \cite{SIEGELMANN199177} famously established that RNNs with rational weights are Turing-complete, while those with finite precision collapse to the class of Regular Languages. This sparked decades of research into extracting Deterministic Finite Automata (DFAs) from trained RNNs using clustering, state-merging, and query-based learning techniques \cite{OMLIN199641,Giles1992,TINO1998171,pmlr-v80-weiss18a}. However, these extraction methods fundamentally rely on the DFA formalism, which suffers from combinatorial state explosion when modeling succinct patterns \cite{hopcroft2001introduction}. More recently, the theoretical focus has shifted to modern architectures; studies have shown that Transformers with hard attention can recognize certain non-regular languages \cite{bhattamishra-etal-2020-ability,Hahn_2020}, while those with soft attention are strictly limited in their capacity to model state-based hierarchy \cite{merrill-2019-sequential,merrill-etal-2020-formal}. Similarly, the resurgence of State-Space Models (SSMs) has renewed interest in continuous-time representations of automata \cite{Gu2021,Dao2022}. Our work diverges from these asymptotic analyses by providing an exact, constructive mapping for AFAs, a class that remains largely unexplored in the context of deep learning despite its dominance in formal verification.

\paragraph{Nondeterminism and Succinctness}
The primary motivation for moving beyond DFAs is succinctness. It is a classical result that Nondeterministic Finite Automata (NFAs) can be exponentially smaller than their equivalent DFAs \cite{Rabin1959}. Early neural approaches attempted to capture this nondeterminism using high-order connections or specialized winner-take-all activations \cite{Casey1996,243123}. This connection is also formalized by proving that Time-Shared Feedforward Networks (TS-FFNs) introduced in~\cite{dhayalkar2025nfa} can efficiently simulate NFAs by exploiting the superposition property of high-dimensional state vectors~\cite{dhayalkar2025nfa,dhayalkar2025neuralnetworksuniversalfinitestate}. While these architectures successfully encoded Existential ($\exists$) logic (thus representing the reachability aspect of nondeterminism), they remained structurally incapable of representing Universal ($\forall$) logic (safety or universality) without an exponential blowup in width. This limitation is significant because AFAs are exponentially more succinct than NFAs and doubly exponentially more succinct than DFAs \cite{Chandra1981,Kozen1976, Fellah01011990}. By introducing learnable logic gates, our current architecture bridges this gap, enabling the network to represent the full boolean spectrum ($\exists$ and $\forall$) directly within the linear recurrence.

\paragraph{Differentiable Logic and Structure Learning}
Our method of learning logic gates via continuous relaxation aligns with the broader field of differentiable logic and neuro-symbolic integration. Frameworks such as Differentiable Inductive Logic Programming~\cite{Evans2018}, Neural Logic Machines \cite{dong2019neurallogicmachines}, and Logic Tensor Networks \cite{serafini2016logictensornetworksdeep,Donadello2017} have demonstrated that logical reasoning can be integrated into gradient-based learning. However, these methods often require specialized, heavy-weight architectures or complex auxiliary losses to enforce consistency. In contrast, our approach relies on the intrinsic geometric properties of standard linear layers. By interpreting the bias term as a sliding threshold between union and intersection, we effectively perform continuous structure learning over the discrete space of boolean circuits \cite{bengio2013estimatingpropagatinggradientsstochastic,petersen2021differentiablesortingnetworksscalable}, similar to relaxation techniques used in Neural Architecture Search (NAS) \cite{liu2018darts,xie2018snas}. This allows us to recover interpretable automata from data without the need for discrete search heuristics.

\paragraph{Automata in Formal Verification}
Finally, our work has implications for neural-guided formal verification. Alternating Automata are the standard formalism for model checking, serving as the operational substrate for Monadic Second-Order Logic (MSO) and Temporal Logics (LTL/CTL) \cite{Vardi1996,Muller1987,Kupferman2000}. While symbolic algorithms for AFA operations (e.g., emptiness, containment) are well-established, they typically rely on symbolic data structures like Binary Decision Diagrams (BDDs) that do not scale well with high-dimensional data. By embedding AFA logic into a differentiable network, our architecture offers a potential pathway for neural model checking, where the network learns to verify properties defined by succinct logical specifications \cite{Wang2018,selsam2018learning, amizadeh2020pdp}. This connects our theoretical results to practical applications in program synthesis and safety verification, where succinct logical representations are paramount.

\section{Preliminaries and Formal Definitions}
\label{sec:preliminaries}

In this section, we formally define the key mathematical concepts used throughout the paper.

\begin{definition}[Alternating Finite Automaton]
\label{DEF:AFA}
An Alternating Finite Automaton (AFA) is a tuple $\mathcal{A} = (Q, \Sigma, \delta, q_0, g, F)$ where:
\begin{itemize}
    \item $Q = \{q_1, \dots, q_n\}$ is a finite set of states.
    \item $\Sigma$ is the finite input alphabet.
    \item $\delta: Q \times (\Sigma \cup \{\varepsilon\}) \rightarrow 2^Q$ is the transition function mapping a state and an input symbol (or the empty string $\varepsilon$) to a set of next states.
    \item $g: Q \rightarrow \{\wedge, \vee\}$ is the labeling function that assigns a logical type to each state: Universal ($\wedge$) or Existential ($\vee$).
    \item $q_0 \in Q$ is the initial state.
    \item $F \subseteq Q$ is the set of accepting (final) states.
\end{itemize}
\end{definition}

AFAs generalize NFAs by introducing two types of states: \textit{existential} states (which accept if at least one transition leads to acceptance) and \textit{universal} states (which accept only if all transitions lead to acceptance). This duality allows AFAs to represent Boolean logic directly within the state transition structure. The transition function may include $\varepsilon$-transitions, which allow the automaton to change state and propagate logic instantaneously without consuming input.

Unlike DFAs or NFAs, acceptance in an AFA is defined recursively. Let $x \in \Sigma^*$ be an input string. We define the predicate $Accept(q, x)$ as follows (assuming implicit handling of instantaneous $\varepsilon$-propagation):

\begin{enumerate}
    \item Base Case ($x = \varepsilon$): The empty string is accepted if and only if $q \in F$ (or can reach $F$ via $\varepsilon$-transitions).
    \item Recursive Step ($x = ay$, where $a \in \Sigma, y \in \Sigma^*$): Let $S = \delta(q, a)$ be the set of successors of $q$ on symbol $a$.
    \begin{itemize}
        \item If $g(q) = \vee$ (Existential), $Accept(q, x)$ is true if $\exists q' \in S$ such that $Accept(q', y)$ is true.
        \item If $g(q) = \wedge$ (Universal), $Accept(q, x)$ is true if $\forall q' \in S$, $Accept(q', y)$ is true.
    \end{itemize}
\end{enumerate}

If the set of successors $S$ is empty, the acceptance value is the identity element of the logic gate: \textit{False} for existential states (rejection) and \textit{True} for universal states (vacuously true). The string $x$ is accepted by the automaton $\mathcal{A}$ if and only if $Accept(q_0, x)$ is true.

\begin{definition}[Boolean State Vector]
\label{def:boolean_state_vector}
Let $Q = \{q_1, \dots, q_n\}$ be the set of states. For a given input prefix $x$, the boolean state vector $v \in \{0,1\}^n$ is defined such that:
\begin{equation}
    [v]_i = 1 \iff \text{State } q_i \text{ is active after processing } x.
\end{equation}
\end{definition}

To map the symbolic AFA semantics onto linear-algebraic operations, this vector captures the instantaneous configuration of the automaton as a binary vector. It allows the network to propagate the state of the entire system forward in time via matrix-vector multiplication, representing the automaton's accumulated logic up to the current step.

\begin{definition}[$\varepsilon$-Closure]
\label{def:epsilon_closure}
Let $Q$ be the state set of a transducer with $\varepsilon$-transitions.
The $\varepsilon$-closure of a state $q \in Q$, denoted $\mathrm{ECL}(q)$,
is the set of all states reachable from $q$ using zero or more
$\varepsilon$-transitions.

For a state set $S \subseteq Q$, the $\varepsilon$-closure is
\[
\mathrm{ECL}(S) = \bigcup_{q \in S} \mathrm{ECL}(q).
\]
\end{definition}

\begin{definition}[Feedforward Network with Activation]
\label{ffn_with_activation}
A feedforward network with depth \( D \) (i.e. \( D \) layers) is a function \( f_\theta : \mathbb{R}^{d} \to \mathbb{R} \) defined as a composition of sigmoid-activated linear transformations:
\[
f_\theta(x) = W_D \cdot \sigma(W_{D-1} \cdot \sigma( \cdots \sigma(W_1 x + b_1) \cdots ) + b_{D-1}) + b_D,
\]
where each \( W_i \in \mathbb{R}^{d_i \times d_{i-1}} \), \( b_i \in \mathbb{R}^{d_i} \) are trainable parameters, and $\sigma$ is an activation function.
\end{definition}

\begin{definition}[Time-Shared, Depth-Unrolled Feedforward Network (TS-FFN)~\cite{dhayalkar2025nfa}]
\label{def:tsfnn}
``A \emph{time-shared, depth-unrolled feedforward network (TS-FFN)} is an feedforward acyclic computation graph obtained by unrolling a shared transition function over the input sequence. Formally, for an input string \( x = (x_1, \dots, x_L) \) with one-hot token encodings \( e_{x_t} \in \mathbb{R}^k \), the network maintains a hidden state \( s_t \in \{0,1\}^n \) and applies the same parameterized transformation at each step:
\[
s_t = \sigma\!\big(T^{x_t} s_{t-1}\big), \quad t = 1,\dots,L,
\]
where \( T^{x_t} \in \{0,1\}^{n \times n} \) are symbol-conditioned transition matrices (shared across all positions) and \( \sigma \) is a threshold activation. The final output \( y = g(s_L) \) is computed by a feedforward readout layer. The term \emph{depth-unrolled} means that each symbol position \(t\) corresponds to one (or a fixed small number of) feedforward layer(s) in the unrolled graph. The term \emph{time-shared} means that the transition parameters are \emph{tied across} all positions \(t\): the same finite set of matrices \(\{T^x\}_{x\in\Sigma}\) (and \(T^\varepsilon\) when present) is reused at every depth, rather than introducing new parameters per step.

Although functionally equivalent to an unrolled recurrent or state-space model (RNN/SSM), the TS-FFN is represented as a purely feedforward, acyclic network with shared parameters across layers, and the TS-FFN viewpoint makes the computation explicitly acyclic and layer-aligned: each depth block can be interpreted as a single symbolic transition (and, when needed, a closure operator), enabling a direct mapping between automata dynamics and network operations.'' \emph{This definition is reproduced verbatim from~\cite{dhayalkar2025nfa}.}
\end{definition}

\section{Theoretical Framework}
\label{SEC:THEORETICAL_FRAMEWORK}
This section formalizes the architectural mechanisms by which nondeterministic state transitions and universal boolean logic are unified within the Time-Shared, Feedforward Network (TS-FFN) framework.

\begin{figure}[t]
\centering
\includegraphics[width=0.5\textwidth]{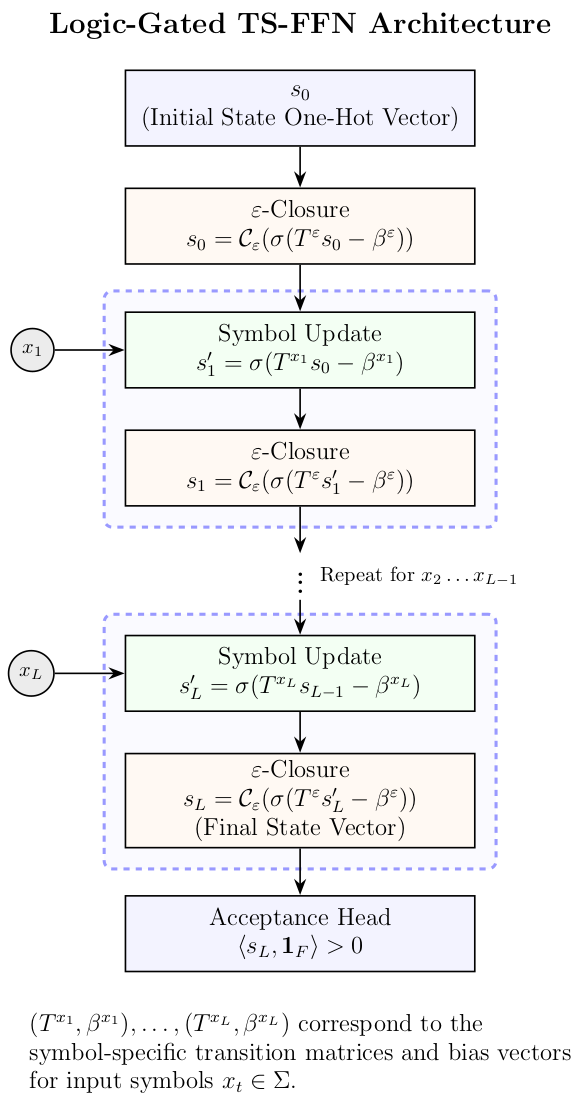}
\caption*{Figure 1: High-level Logic-Gated TS-FFN visualization. An Logic-Gated TS-FFN unrolls a shared logic-gated transition rule across depth (one block per symbol position). Each step applies a symbol-specific transition matrix $T^{x_t}$ and bias vector $\beta^{x_t}$ (with parameter pairs $(T^{x_1}, \beta^{x_1}), \dots, (T^{x_L}, \beta^{x_L})$ corresponding to input symbols from $\Sigma$), time-shared across $t$, interleaved with an $\varepsilon$-closure operator $\mathcal{C}_\varepsilon$ (parameterized by $T^\varepsilon, \beta^\varepsilon$). The resulting boolean state vector $s_t$ corresponds to the AFA's active configuration after consuming prefix $x_{1:t}$ and applying $\varepsilon$-closure, effectively implementing both Existential (\textsc{OR}) and Universal (\textsc{AND}) branching logic within the recurrence.}
\label{FIG:TSFFN-VIS}
\end{figure}

\subsection{Logic-Gated Neural Architecture}
\label{subsec:logic_gated_neural_architecture}
\begin{definition}[Logic-Gated Time-Shared Feedforward Network]
\label{DEF:LOGIC_GATED_TSFFN}
A Logic-Gated TS-FFN is a depth-unrolled network where each layer $t$ applies a linear transformation followed by a state-dependent threshold activation. The architecture processes the input sequence in chronological order (from $t=1 \to L$) to propagate state configurations.

The update rule at step $t$ (corresponding to input symbol $x_t$) incorporates both symbol consumption and instantaneous $\varepsilon$-logic propagation:
\begin{equation}
    v_t = \mathcal{C}_\varepsilon \left( \sigma\left(T^{(x_t)} v_{t-1} - \beta^{(x_t)}\right) \right)
\end{equation}
where:
\begin{itemize}
    \item $v_t \in \{0,1\}^n$ is the boolean state vector representing the configuration after processing the prefix $x_1\dots x_t$.
    \item $T^{(x_t)}$ and $\beta^{(x_t)}$ are the transition matrix and logic-gating bias vector for symbol $x_t$.
    \item $\mathcal{C}_\varepsilon(\cdot)$ is the $\varepsilon$-closure operator, defined as the boolean fixed point of the recurrent update $u \leftarrow \sigma(T^{(\varepsilon)} u - \beta^{(\varepsilon)})$, which resolves instantaneous logic propagation across $\varepsilon$-transitions. This operation is applied in the same procedure as mentioned in~\cite{dhayalkar2025nfa}
    \item $\sigma(z) = \mathbf{1}_{[z \ge 0]}$ is the binary step activation function.
\end{itemize}
\end{definition}

Standard TS-FFNs, as used for NFA simulation \cite{dhayalkar2025nfa}, employ a fixed binary threshold activation $\sigma(Wx)$. In a binary setting where weights $W_{ij} \in \{0,1\}$, the pre-activation sum $z = \sum W_{ij} v_j$ simply counts the number of active predecessors. A binary threshold ($\sigma(z) = \mathbf{1}_{[z > 0]}$) triggers activation if this count is at least 1, which corresponds strictly to a logical \textsc{OR}: a neuron fires if the weighted sum of its inputs is strictly positive (i.e., if \textit{at least one} incoming path is active). While sufficient for nondeterminism, this fixed threshold cannot represent the Universal \textsc{AND} branching inherent to AFAs, which requires a stricter condition: a state must activate only if \textit{all} of its $d_{in}$ incoming paths are active simultaneously. A standard zero-threshold neuron cannot distinguish between ``1 active input'' and ``$d_{in}$ active inputs'' and instead it fires for both.

The learnable bias term $\beta$ solves this by acting as a negative offset to the threshold. The activation condition becomes:
\begin{equation}
    \sum_{j} T_{ij} v_j - \beta > 0 \implies \text{Active inputs} > \beta
\end{equation}
By tuning $\beta$, we can precisely control the evidence required to fire:
\begin{itemize}
    \item Setting $\beta \approx 0.5$ requires strictly $>0$ inputs (i.e., $\ge 1$), recovering existential logic.
    \item Setting $\beta \approx d_{in} - 0.5$ requires strictly $> d_{in} - 1$ inputs (i.e., exactly $d_{in}$), achieving universal logic.
\end{itemize}
Thus, $\beta$ transforms the neuron from a static \textsc{OR} gate into a configurable logic gate capable of representing the full boolean spectrum required for Alternating Automata. Refer to Fig.~\hyperref[FIG:TSFFN-VIS]{1} for a high-level visualization.

\begin{proposition}[Logical Aggregation via Biased Linear Units]
\label{PROP:LOGICAL_AGGREGATION}
Let $q_i \in Q$ be a state with a set of predecessors $P_i = \{q_j \mid q_i \in \delta(q_j, x)\}$ for input symbol $x \in \Sigma \cup \{\varepsilon\}$. Let $d_i = |P_i|$ denote the in-degree of state $q_i$ for this symbol (the number of incoming connections). Let $v_{prev} \in \{0,1\}^n$ be the boolean activation vector of the states at the previous time step.

The boolean activation value of $q_i$ at the current step can be computed exactly via a linear threshold operation:
\begin{equation}
    v_{curr}[i] = \sigma \left( \sum_{j=1}^n T_{ij}^{(x)} v_{prev}[j] - \beta_i^{(x)} \right)
\end{equation}
where $T_{ij}^{(x)} = 1 \iff q_i \in \delta(q_j, x)$ (encoding the transition $q_j \xrightarrow{x} q_i$), and the state-dependent bias $\beta_i^{(x)}$ is defined as:
\begin{itemize}
    \item Case 1 (Existential Aggregation / \textsc{OR}): $\beta_i^{(x)} = 0.5$.
    \item Case 2 (Universal Aggregation / \textsc{AND}): $\beta_i^{(x)} = d_i - 0.5$.
\end{itemize}
\end{proposition}

\emph{Proof Sketch provided in Appendix~\ref{PROOF:PROP:LOGICAL_AGGREGATION}}. This proposition demonstrates that the structural distinction between collecting signals via Union (\textsc{OR}) or Intersection (\textsc{AND}) can be encoded entirely within the bias vector of a standard linear layer. The bias $\beta$ effectively sets the firing threshold of the neuron relative to its incoming connectivity $d_i$. A low threshold ($\beta \approx 0$) creates an \textsc{OR} gate, allowing the state to activate if any path reaches it. A high threshold ($\beta \approx d_i$) creates an \textsc{AND} gate, requiring all incoming paths to be active. This is also explained in the above Defintion~\ref{DEF:LOGIC_GATED_TSFFN} and allows the Forward TS-FFN to function as a depth-unrolled boolean circuit, applying logic uniformly to both input-driven transitions and instantaneous $\varepsilon$-closure steps defined in Definition~\ref{DEF:LOGIC_GATED_TSFFN}.

Relation to Prior Work:
This generalizes the binary state vector representation for NFAs introduced in \cite{dhayalkar2025nfa}. In that framework, the bias was implicitly zero, restricting the network to calculating reachability (Existential quantification). By explicitly parameterizing the bias, we extend the representational capacity of the network to handle complex boolean aggregations (such as requiring multiple simultaneous paths to converge), unlocking greater expressivity within the fundamental linear-algebraic structure of the update rule.

\begin{theorem}[Simulation of AFAs by Logic-Gated TS-FFNs]
\label{THM:AFA_SIMULATION}
Let $\mathcal{A} = (Q, \Sigma \cup \{\varepsilon\}, \delta, q_0, g, F)$ be an Alternating Finite Automaton with $n$ states. There exists a time-shared, depth-unrolled feedforward network $f_\theta$ that simulates the state transitions of $\mathcal{A}$. Specifically, for any input string $x = x_1 \dots x_L$, the network computes the state configuration through the following forward-unrolled process:

\begin{enumerate}
    \item Initialization ($t=0$): Let $e_{q_0} \in \{0,1\}^n$ be the one-hot indicator vector for the start state $q_0$ (where $[e_{q_0}]_i = 1$ if $q_i = q_0$, and $0$ otherwise). The boolean state vector $v_0$ is initialized to this vector, followed by an instantaneous $\varepsilon$-closure to account for initial $\varepsilon$-transitions:
    \begin{equation}
        v_0 = \mathcal{C}_\varepsilon(e_{q_0})
    \end{equation}
    This encodes the base case: at the start of the string, the active states are $q_0$ and any states reachable via $\varepsilon$-transitions.

    \item Forward Propagation ($t = 1, \dots, L$): For each symbol $x_t$ in the sequence (processed in chronological order), the network updates the state vector via the Logic-Gated layer defined in Definition \ref{DEF:LOGIC_GATED_TSFFN}, incorporating the $\varepsilon$-closure operator $\mathcal{C}_\varepsilon$:
    \begin{equation}
        v_t = \mathcal{C}_\varepsilon \left( \sigma\left(T^{(x_t)} v_{t-1} - \beta^{(x_t)}\right) \right)
    \end{equation}
    where $T^{(x_t)}$ is the transition matrix for symbol $x_t$, and $\beta^{(x_t)}$ is the connectivity-dependent bias vector defined in Proposition \ref{PROP:LOGICAL_AGGREGATION}.

    \item Readout ($t=L$): The input string $x$ is accepted if the final state configuration $v_L$ satisfies the acceptance condition (typically, having a non-empty intersection with the set of final states $F$):
    \begin{equation}
        f_\theta(x) = \mathbf{1}_{[v_L \cdot \mathbf{1}_F^T > 0]}
    \end{equation}
\end{enumerate}
\end{theorem}

\emph{Proof Sketch provided in Appendix~\ref{PROOF:THM:AFA_SIMULATION}}. A critical distinction of this architecture is the incorporation of logic-gated aggregation within the forward pass. In standard NFA simulations such as in~\cite{dhayalkar2025nfa}, the network propagates reachability using fixed \textsc{OR}-logic ($t=0 \to L$). Here, the Logic-Gated TS-FFN maintains the forward processing order but enhances the state update dynamics. By allowing neurons to implement both \textsc{OR} (existential) and \textsc{AND} (universal) aggregation via the learned bias $\beta$, the network can model complex synchronization requirements and universal quantification constraints directly within the forward unrolling, effectively extending the computational capacity beyond standard NFA simulation without increasing the state space size $n$.

Relation to Prior Work:
This construction generalizes the TS-FFN framework to handle complex boolean dependencies. While \cite{dhayalkar2025nfa} established that TS-FFNs can simulate nondeterminism via thresholding, that result was limited to Existential ($\exists$) aggregation. Theorem \ref{THM:AFA_SIMULATION} proves that the exact same backbone, when equipped with appropriate biases, is capable of simulating Universal ($\forall$) aggregation, allowing the network to enforce that multiple computational paths must remain active simultaneously to trigger a transition.

\subsection{Succinctness and Complexity}
\label{sec:succinctness}

We now analyze the resource complexity of the Logic-Gated TS-FFN. A key advantage of Alternating Finite Automata is their ability to represent certain regular languages exponentially more succinctly than NFAs or DFAs. We show that our neural construction inherits this succinctness directly.

\begin{proposition}[Succinctness and Parameter Efficiency]
\label{PROP:SUCCINCTNESS}
Let $\mathcal{A}$ be an AFA with $n$ states and an alphabet of size $|\Sigma|=k$. Construct a Logic-Gated TS-FFN $f_\theta$ that simulates $\mathcal{A}$ according to Theorem \ref{THM:AFA_SIMULATION}. Then:
\begin{enumerate}
    \item Parameter Count: The total number of trainable parameters is $\mathcal{O}(k n^2)$. Specifically, the model requires:
    \begin{itemize}
        \item $k$ symbol transition matrices $T^{(x)}$ and 1 $\varepsilon$-transition matrix $T^{(\varepsilon)}$, each of size $n \times n$.
        \item $k$ symbol bias vectors $\beta^{(x)}$ and 1 $\varepsilon$-bias vector $\beta^{(\varepsilon)}$, each of size $n$.
    \end{itemize}
    This count depends only on the automaton size and alphabet, and is strictly independent of the input string length $L$.
    
    \item Succinctness: The network width corresponds exactly to the number of states $n$. Since an AFA with $n$ states can recognize languages requiring $2^n$ states in an NFA, this architecture achieves exponentially greater representational capacity per neuron compared to standard NFA-simulating networks.
\end{enumerate}
\end{proposition}

\emph{Proof Sketch provided in Appendix~\ref{PROOF:PROP:SUCCINCTNESS}}. This result highlights a fundamental efficiency gain: we achieve exponential expressivity not by adding more neurons, but by enriching the \textit{interactions} between them. The addition of the learnable bias vector $\beta^{(x)}$ transforms a standard linear layer into a dynamic logic circuit. This allows the network to represent complex boolean conditions (intersections and complements of path properties) essentially for free within the linear update step. Consequently, the model can capture patterns that would otherwise require massive over-parameterization in standard recurrent or feedforward models. Furthermore, since the transition matrices and bias vectors are time-shared, the total parameter count remains $\mathcal{O}(kn^2)$ and is \textbf{strictly independent of the input sequence length $L$}.

\begin{remark}[Comparison to NFA Simulation]
Prior research \cite{dhayalkar2025nfa} demonstrated that TS-FFNs can simulate NFAs with $\mathcal{O}(kn^2)$ parameters. The Logic-Gated architecture maintains this same parameter complexity class while dramatically increasing expressivity. While the NFA-based networks described in \cite{dhayalkar2025nfa} are restricted to languages with compact NFA representations, the Logic-Gated architecture can efficiently learn languages that are structurally complex (e.g., those requiring universal quantification $\forall$), which standard NFA networks can only approximate via much larger state spaces.
\end{remark}

\subsection{Encoding input strings within the Logic-Gated TS-FFN framework}
A central component of the Logic-Gated Time-Shared Feedforward Network (TS-FFN) construction is how an input string $x = x_1 x_2 \dots x_L \in \Sigma^*$ is provided to the network. Unlike conventional neural networks that embed an entire string into a single continuous vector, the Logic-Gated TS-FFN decouples the symbolic sequence from the network architecture itself. The string instead determines the sequence of shared transition matrices and logic-gated bias vectors that are applied during the forward pass, producing an acyclic unrolled computation that exactly mirrors the logical transition semantics of an Alternating Finite Automaton.

Transition parameter selection:
For every symbol $x \in \Sigma$, the model maintains a corresponding pair consisting of a symbolic transition matrix $T^{(x)}$ and a logic-gating bias vector $\beta^{(x)}$. These parameter pairs encode both the connectivity and the boolean logic (AND/OR) associated with transitions on $x$. They can either be:
\begin{itemize}
    \item Fixed (Symbolic): Directly encoded from the known transition function $\delta$ and labeling function $g$ of the target AFA, as defined in Definition~\ref{DEF:LOGIC_GATED_TSFFN}.
    \item Learned (Parameterized): Initialized randomly and optimized via gradient descent using labeled examples, allowing the network to jointly learn topology and logical types as discussed in Proposition~\ref{PROP:LEARNABILITY}.
\end{itemize}

Encoding the string:
To process the input string $x = x_1 x_2 \dots x_L$, the TS-FFN does not embed $x$ as a single vector input. Instead, the symbols $x_t$ dictate which shared parameters are active at each depth $t$ of the unrolled graph. Specifically, the presence of symbol $x_t$ at position $t$ triggers the application of the specific tuple $(T^{(x_t)}, \beta^{(x_t)})$ to the hidden state $v_{t-1}$. This mechanism ensures that the network dynamically reconfigures its layer-wise boolean logic to match the input sequence, supporting both symbol-driven transitions and instantaneous $\varepsilon$-closures ($\mathcal{C}_\varepsilon$) at every step.

An example of encoding input strings into the network is provided in Appendix~\ref{EXAMPLE}

\section{Equivalence and Learnability}
\label{SEC:LEARNABILITY}

We now address the theoretical equivalence between our neural architecture and the class of Alternating Finite Automata.

\begin{theorem}[Equivalence between Logic-Gated TS-FFNs and AFAs]
\label{THM:EQUIVALENCE}
Let $\mathcal{L} \subseteq \Sigma^*$ be any regular language. Then:
\begin{enumerate}
    \item (Forward Direction) There exists a symbolic Logic-Gated Time-Shared, Depth-Unrolled Feedforward Network $f_{\theta}$ with:
    \begin{itemize}
        \item transition matrices $\{T^{(x)}\}_{x \in \Sigma \cup \{\varepsilon\}}$ and bias vectors $\{\beta^{(x)}\}_{x \in \Sigma \cup \{\varepsilon\}}$ representing the AFA structure,
        \item parameter sharing across all time steps (depth positions),
        \item width $\mathcal{O}(n)$, where $n$ is the number of AFA states,
    \end{itemize}
    such that for any input string $x \in \Sigma^*$, the network accepts $x$ if and only if $x \in \mathcal{L}$.

    \item (Reverse Direction) Every Logic-Gated Time-Shared, Depth-Unrolled Feedforward Network constructed using this symbolic simulation procedure corresponds to an AFA $\mathcal{A}'$ such that, for all $x \in \Sigma^*$:
    \begin{equation}
        f_{\theta}(x) = 1 \iff x \in \mathcal{L}(\mathcal{A}').
    \end{equation}
\end{enumerate}
\end{theorem}

\emph{Proof Sketch provided in Appendix~\ref{PROOF:THM:EQUIVALENCE}}. This theorem establishes that the Logic-Gated TS-FFN is not merely an approximation of an automaton, but a rigorous algebraic isomorphism of the AFA formalism. The forward result guarantees that our architecture possesses the necessary expressivity to represent any regular pattern efficiently. The reverse result guarantees that the network cannot hallucinate non-regular behaviors, ensuring that trained models remain verifiable and interpretable.

\begin{proposition}[Gradient-Based Learnability of Logic Gates]
\label{PROP:LEARNABILITY}
Let $\mathcal{D} = \{(x^{(i)}, y^{(i)})\}$ be a dataset of input strings $x^{(i)}$ labeled by the acceptance $y^{(i)}$ of an unknown target AFA. A Logic-Gated TS-FFN, parameterized by trainable symbol matrices $\{T^{(x)}, \beta^{(x)}\}$ and epsilon matrices $\{T^{(\varepsilon)}, \beta^{(\varepsilon)}\}$, can be trained to approximate the target AFA via gradient descent.

To ensure differentiability during training, we relax the discrete constraints as follows:
\begin{enumerate}
    \item The binary threshold activation $\sigma(z) = \mathbf{1}_{[z \ge 0]}$ is relaxed to the logistic sigmoid function $\sigma_\lambda(z) = (1 + e^{-\lambda z})^{-1}$, where $\lambda$ is a temperature parameter controlling the steepness of the transition.
    \item The binary constraints on the transition matrices are relaxed to $T^{(x)}, T^{(\varepsilon)} \in \mathbb{R}^{n \times n}$.
    \item The bias vectors $\beta^{(x)}, \beta^{(\varepsilon)} \in \mathbb{R}^n$ are learned freely without constraints.
\end{enumerate}

During the optimization process:
\begin{itemize}
    \item The transition matrices learn the connectivity structure (adjacency) of both symbol-driven and instantaneous $\varepsilon$-transitions.
    \item The bias vectors learn the logical type of each state. A learned bias value $\beta \approx 0.5$ drives the node towards Existential (\textsc{OR}) aggregation, while a value $\beta \approx \text{in-degree} - 0.5$ drives it towards Universal (\textsc{AND}) aggregation.
    \item The forward-unrolled computation graph accumulates gradients from the entire sequence length, enabling the network to infer the appropriate logical structure solely from end-to-end binary acceptance labels via standard backpropagation.
\end{itemize}
\end{proposition}

\emph{Proof Sketch provided in Appendix~\ref{PROOF:PROP:LEARNABILITY}}. The significance of this result is that the network does not merely memorize patterns; it effectively performs a differentiable search over the space of boolean logic circuits. The parameter $\beta$ acts as a continuous ``slider'' between \textsc{OR} and \textsc{AND} logic. By allowing this parameter to evolve via gradient descent, the network can dynamically discover whether a specific hidden state should act as a union or an intersection of its incoming paths. This avoids the combinatorial complexity of discrete structure learning algorithms, replacing them with a smooth optimization landscape where the ``logic gate type'' is a learned feature of the neuron.

Relation to Prior Work:
This extends the learnability results for NFAs and Probabilistic Finite Automatas presented in \cite{dhayalkar2025nfa,dhayalkar2025symbolicfeedforwardnetworksprobabilistic}. In those works, learnability was restricted to identifying edges ($T^{(x)}$) because the activation threshold was fixed (implicitly representing only \textsc{OR} gates). Here, we show that by relaxing the threshold parameter $\beta$, we can simultaneously learn both the graph topology \textit{and} the logical semantics of the nodes. This suggests that the Logic-Gated architecture is a more general-purpose symbolic learner, capable of recovering a broader class of automata from data without prior knowledge of their structural type.

\begin{remark}[On ``Soft'' Logic Gates]
\label{rem:soft_logic}
During the training phase, the learned bias parameters $\beta$ may take continuous values rather than discrete integers. A bias value in an intermediate range, such as $\beta \in (0.5, 1.5)$, effectively implements a ``soft \textsc{AND}'' or a ``weighted majority'' gate, requiring a specific fraction of active inputs to trigger the activation. This continuous relaxation allows the gradient descent process to smoothly traverse the space of boolean logic gates, avoiding the combinatorial hardness associated with discrete structure learning. The network can thus gradually transition a node's behavior from Existential \textsc{OR} to Universal (\textsc{AND}) as evidence accumulates during training.
\end{remark}

\section{Experiments}
\label{sec:experiments}

\subsection{Experimental Setup}
\label{subsec:exp_setup}

To empirically validate the theoretical results established in Sections \ref{SEC:THEORETICAL_FRAMEWORK} and \ref{SEC:LEARNABILITY}, we design a controlled experimental framework that jointly simulates Alternating Finite Automata (AFAs) and their corresponding Logic-Gated Time-Shared Feedforward Networks (TS-FFNs). All experiments are implemented in PyTorch \cite{pytorch} and executed on an NVIDIA GeForce RTX 4060 GPU.

\paragraph{Synthetic AFA Generation.}
We evaluate our framework under two distinct complexity configurations:
\begin{itemize}
    \item \textbf{Configuration 1 (Baseline):} Each instance contains $n=20$ states with input alphabet $\Sigma=\{a,b,c,d,e,f\}$ (hence $|\Sigma|=6$). Input strings are uniformly sampled with lengths between 1 and 50.
    \item \textbf{Configuration 2 (High-Complexity):} Each instance contains $n=1,000$ states with a larger input alphabet $\Sigma=\{a,b,...,z,A,B,...Z,0,1,...,9\}$ (hence $|\Sigma|=62$). Input strings are uniformly sampled with lengths between 1 and 2,000.
\end{itemize}

For each random seed, we generate a fresh symbolic AFA. State types (Universal/Existential) are assigned randomly, with $q_i \in Q_{\exists}$ or $q_i \in Q_{\forall}$ with equal probability. Transition targets are sampled uniformly from the state space $Q$. To rigorously test the $\varepsilon$-closure mechanisms formalized in Theorem \ref{THM:AFA_SIMULATION}, we inject $\varepsilon$-transitions between random state pairs with probability $p_\varepsilon=0.3$. The start state is fixed as $q_0=0$, and the set of accepting states $F$ is a random subset of size $\lfloor n/2 \rfloor$.

\paragraph{Dataset Construction.}
For each generated AFA, we construct synthetic datasets by uniformly sampling input strings from $\Sigma^*$.
\begin{itemize}
    \item For Configuration 1, we generate 5,000 training samples and 500 test samples per seed.
    \item For Configuration 2, we generate 100,000 training samples and 10,000 test samples per seed.
\end{itemize}
Ground-truth acceptance labels $y \in \{0,1\}$ are computed by direct symbolic execution of the AFA, explicitly resolving $\varepsilon$-closures and recursive boolean conditions as defined in Definition\ref{DEF:AFA}.

\paragraph{Network Architecture and Training}
Each experiment instantiates a Logic-Gated TS-FFN corresponding to the target AFA configuration.
\begin{itemize}
    \item For Simulation Experiments (Section \ref{EXP:PROP:LOGICAL_AGGREGATION} to ~\ref{EXP:THM:EQUIVALENCE}): The network parameters (matrices $T^{(x)}$ and biases $\beta^{(x)}$) are initialized directly from the symbolic definition of the AFA according to Theorem \ref{THM:AFA_SIMULATION}.
    \item For Learnability Experiments (Section \ref{EXP:PROP:LEARNABILITY}): The parameters are initialized randomly using Kaiming initialization \cite{7410480}. The model is trained via gradient descent using the Adam optimizer~\cite{kingma2017adammethodstochasticoptimization} with a learning rate of $0.01$ and binary cross-entropy loss. To enable differentiability, we relax the binary threshold to a sigmoid activation during training, as described in Proposition \ref{PROP:LEARNABILITY}.
\end{itemize}
All reported results are averaged over five independent random seeds, with confidence intervals computed using Student's $t$-distribution.

\begin{table}[ht]
\centering
\caption*{Table I: Validation results for experiments across both configurations. Accuracy and confidence intervals computed over 5 random seeds using Student's $t$-distribution.}
\label{TABLE:RESULTS}
\vspace{2mm}
\setlength{\tabcolsep}{1pt} 
\renewcommand{\arraystretch}{1.1}  
\begin{tabular}{lcccc}
\textbf{Validation} & \textbf{Config} & \textbf{Mean} & \textbf{Std Dev} & \textbf{95\% CI} \\
\textbf{Experiment} & & \textbf{Accuracy} & & \\
\hline
~\ref{EXP:PROP:LOGICAL_AGGREGATION}: Proposition~\ref{PROP:LOGICAL_AGGREGATION} & 1 & 1.0000 & 0.0000 & (1.0000, 1.0000) \\
~\ref{EXP:PROP:LOGICAL_AGGREGATION}: Proposition~\ref{PROP:LOGICAL_AGGREGATION} & 2 & 1.0000 & 0.0000 & (1.0000, 1.0000) \\
\hline
~\ref{EXP:THM:AFA_SIMULATION}: Theorem~\ref{THM:AFA_SIMULATION} & 1 & 1.0000 & 0.0000 & (1.0000, 1.0000) \\
~\ref{EXP:THM:AFA_SIMULATION}: Theorem~\ref{THM:AFA_SIMULATION} & 2 & 1.0000 & 0.0000 & (1.0000, 1.0000) \\
\hline
~\ref{EXP:THM:EQUIVALENCE}: Theorem~\ref{THM:EQUIVALENCE} & 1 & 1.0000 & 0.0000 & (1.0000, 1.0000) \\
\qquad \ (forward direction) & & & & \\
~\ref{EXP:THM:EQUIVALENCE}: Theorem~\ref{THM:EQUIVALENCE} & 2 & 1.0000 & 0.0000 & (1.0000, 1.0000) \\
\qquad \ (forward direction) & & & & \\
~\ref{EXP:THM:EQUIVALENCE}: Theorem~\ref{THM:EQUIVALENCE} & 1 & 1.0000 & 0.0000 & (1.0000, 1.0000) \\
\qquad \ (backward direction) & & & & \\
~\ref{EXP:THM:EQUIVALENCE}: Theorem~\ref{THM:EQUIVALENCE} & 2 & 1.0000 & 0.0000 & (1.0000, 1.0000) \\
\qquad \ (backward direction) & & & & \\
\hline
~\ref{EXP:PROP:LEARNABILITY}: Proposition~\ref{PROP:LEARNABILITY} & 1 & 1.0000 & 0.0000 & (1.0000, 1.0000) \\
~\ref{EXP:PROP:LEARNABILITY}: Proposition~\ref{PROP:LEARNABILITY} & 2 & 0.9993 & 0.00013 & (0.99909, 0.99934) \\
\end{tabular}
\end{table}

\subsection{Validation of Proposition \ref{PROP:LOGICAL_AGGREGATION}: Logical Aggregation}
\label{EXP:PROP:LOGICAL_AGGREGATION}

To empirically validate Proposition \ref{PROP:LOGICAL_AGGREGATION}, we verify whether the state-dependent biased linear units correctly implement the required boolean logic (Existential/\textsc{OR} vs. Universal/\textsc{AND}). For every neuron in the generated networks across all seeds, we perform an exhaustive deterministic check: we calculate the neuron's activation for every possible number of active predecessors $k \in \{0, \dots, d_{in}\}$.

As summarized in Table~\hyperref[TABLE:RESULTS]{I}, the neural logic gates matched the symbolic ground truth with perfect accuracy (100\%) in both configurations. This empirically confirms that the bias parameter $\beta$ successfully modulates the linear threshold to switch between Existential ($\beta \approx 0.5$) and Universal ($\beta \approx d_{in} - 0.5$) modes, validating that boolean logic can be encoded directly within the linear weights and biases of the Forward TS-FFN.

\subsection{Validation of Theorem \ref{THM:AFA_SIMULATION}: Exact Simulation}
\label{EXP:THM:AFA_SIMULATION}

To validate Theorem \ref{THM:AFA_SIMULATION}, we evaluate whether the Logic-Gated TS-FFN can simulate the complete state evolution of a randomly generated AFA, including the complex dynamics of $\varepsilon$-closures and universal branching. For each random seed, we generate a distinct AFA and construct the corresponding neural network by directly setting its weights $T^{(x)}$ and biases $\beta^{(x)}$ according to the forward construction in Theorem \ref{THM:AFA_SIMULATION}, without any training.

We evaluate the network on the sampled test strings per seed and compare its binary acceptance decision against the ground-truth AFA execution. As shown in Table~\hyperref[TABLE:RESULTS]{I}, the network achieved 100\% simulation accuracy across all seeds and configurations. This empirically confirms that the Logic-Gated TS-FFN, equipped with the monotonic $\varepsilon$-closure operator, performs an exact simulation of the AFA's forward reachability semantics.

\subsection{Validation of Proposition \ref{PROP:SUCCINCTNESS}: Succinctness and Efficiency}
\label{EXP:PROP:SUCCINCTNESS}

We validate Proposition \ref{PROP:SUCCINCTNESS} regarding the parameter efficiency and representational capacity of the architecture. We compare the actual parameter count of the Logic-Gated TS-FFN against the derived theoretical bound $\mathcal{O}(kn^2)$. Furthermore, we quantify the \textit{succinctness ratio}, defined as the ratio between the state space size of an equivalent NFA ($2^n$) and the width of our network ($n$).

\begin{table}[h]
\centering
\caption*{Table II: Parameter Efficiency and Succinctness. The Logic-Gated TS-FFN maintains a parameter count of $\mathcal{O}(kn^2)$ while achieving an exponential succinctness ratio ($2^n/n$) compared to equivalent NFA representations.}
\setlength{\tabcolsep}{4pt} 
\label{TAB:SUCCINCTNESS_RESULTS}
\begin{tabular}{lcccc}
\textbf{Config} & \textbf{Width} & \textbf{NFA Size} & \textbf{Params} & \textbf{Succinctness} \\
& ($n$) & ($2^n$) & & \textbf{Ratio} \\
\midrule
1 & 20 & $1.05 \times 10^6$ & 2,960 & $5.24 \times 10^4$ \\
2 & 200 & $1.61 \times 10^{60}$ & 2,532,800 & $8.03 \times 10^{57}$ \\
\end{tabular}
\end{table}

The results, averaged over five seeds, are presented in Table~\hyperref[TAB:SUCCINCTNESS_RESULTS]{II}.
\begin{itemize}
    \item Parameter Efficiency: In both configurations, the actual parameter counts matched the theoretical expectations exactly (2960 parameters for Configuration 1, 2532800 parameters for Configuration 2, confirming the quadratic $\mathcal{O}(kn^2)$ scaling.
    \item Succinctness: The architecture demonstrates massive spatial efficiency. In Configuration 1, the network achieves a succinctness ratio of $5.24 \times 10^4$ times compared to a standard NFA representation. In Configuration 2, this ratio grows to over $8.03 \times 10^{57}$.
\end{itemize}
These results empirically verify that the Logic-Gated architecture achieves exponential state compression, representing complex regular languages with a network width linear in the AFA size rather than exponential.

\subsection{Validation of Theorem \ref{THM:EQUIVALENCE}: Bidirectional Equivalence}
\label{EXP:THM:EQUIVALENCE}

To empirically validate Theorem \ref{THM:EQUIVALENCE}, we test the structural isomorphism between the neural and symbolic representations in both directions.

Forward Direction:
For each random seed, we generate a random AFA with $\varepsilon$-transitions and construct the corresponding Logic-Gated TS-FFN using the forward simulation mapping defined in Theorem \ref{THM:AFA_SIMULATION}. We then evaluate the network on the test strings per seed. As shown in Table~\hyperref[TABLE:RESULTS]{I}, the network matched the ground-truth AFA decisions with 100\% accuracy across all trials. This confirms that the constructive mapping preserves the accepted language exactly, effectively embedding the automaton's logic into the network's forward recurrence.

Reverse Direction:
To validate the reverse implication, we instantiate a Logic-Gated TS-FFN (initialized from a random AFA to ensure validity) and apply the extraction procedure described in the proof of Theorem \ref{THM:EQUIVALENCE}. This procedure recovers the symbolic transition function $\delta$ and logical types $g$ directly from the network's trainable parameters:
\begin{itemize}
    \item The transition structure is recovered by thresholding the forward weight matrices $T^{(x)}$ and $T^{(\varepsilon)}$.
    \item The logical types (\textsc{AND}/\textsc{OR}) are inferred from the bias vectors $\beta^{(x)}$, distinguishing between existential ($\beta \approx 0.5$) and universal ($\beta \approx d_{in}-0.5$) aggregation.
\end{itemize}
We then execute the extracted AFA on the same test set. The results in Table~\hyperref[TABLE:RESULTS]{I} show perfect agreement between the original AFA's behavior and the extracted automaton. This confirms that the neural architecture is not merely an approximation but a faithful algebraic encoding, capable of being losslessly converted back into symbolic form.

\subsection{Validation of Proposition \ref{PROP:LEARNABILITY}: Gradient-Based Learnability}
\label{EXP:PROP:LEARNABILITY}

To validate Proposition \ref{PROP:LEARNABILITY}, we assess whether the Logic-Gated TS-FFN can learn an unknown target AFA solely from input-output examples $\{(x^{(i)}, y^{(i)})\}$. Unlike the simulation experiments where parameters were fixed, here we initialize the network with random weights (Kaiming uniform~\cite{7410480}) and train it via gradient descent to minimize binary cross-entropy loss. We train the model using the Adam optimizer~\cite{kingma2017adammethodstochasticoptimization} with a learning rate of $0.01$, a batch size of 64, and apply gradient clipping (norm 1.0) on datasets of 5,000 (Configuration 1) and 20,000 (Configuration 2) samples.

To enable differentiability for discrete boolean logic, we employ a three-phase curriculum learning strategy over 20 epochs. We relax the binary threshold activation to a logistic sigmoid function $\sigma_\lambda(z)$ and gradually increase the temperature $\lambda$ and maximum sequence length $L$ across phases ($L=10, \lambda=1.0$ for 5 epochs; $L=25, \lambda=2.5$ for 5 epochs; $L=50, \lambda=5.0$ for 10 epochs). Additionally, the bias parameters $\beta$ are initialized to $1.5$ to encourage sparse activation start states. This allows gradients to propagate through the unrolled computation graph, enabling the network to simultaneously learn the connectivity structure ($T$) and the logical types ($\beta$) of the hidden states from end-to-end supervision.

We evaluated the learning performance on datasets generated from hidden target AFAs for both configurations.
\begin{itemize}
    \item Configuration 1 (Baseline): The network achieved perfect test accuracy (100\%) across all random seeds, successfully recovering the exact acceptance language of the target AFA.
    \item Configuration 2 (High-Complexity): The network achieved near-perfect accuracy ($>99\%$), demonstrating that the gradient-based approach scales to larger state spaces and alphabets even when complex $\varepsilon$-closures and universal branching are present.
\end{itemize}

\begin{figure}[t]
    \centering
    \includegraphics[width=0.75\linewidth]{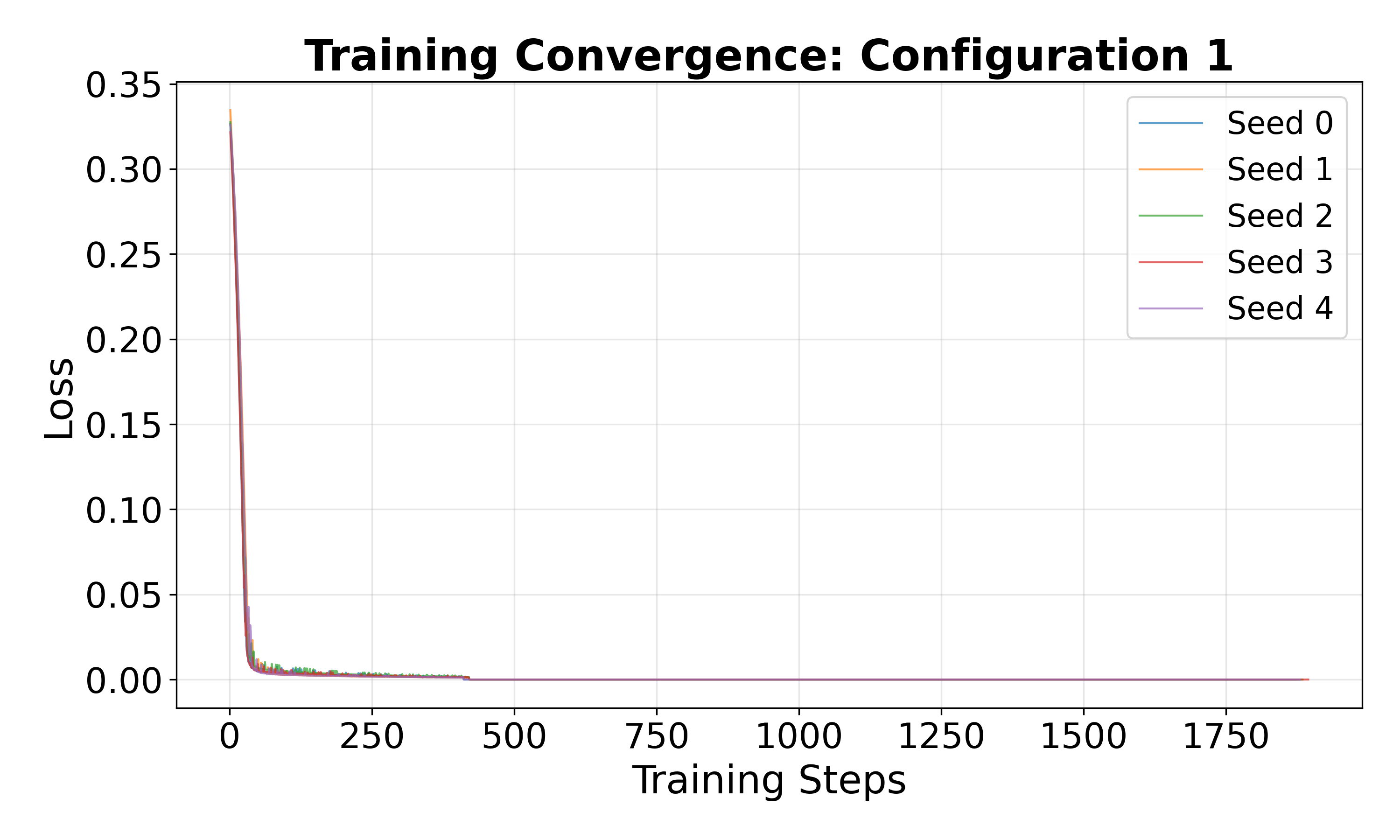}
    \caption*{Figure 2: Per-batch training loss across 5 seeds for Configuration 1.}
    \label{FIG:TRAIN_CONFIG1}
\end{figure}

\begin{figure}[t]
    \centering
    \includegraphics[width=0.75\linewidth]{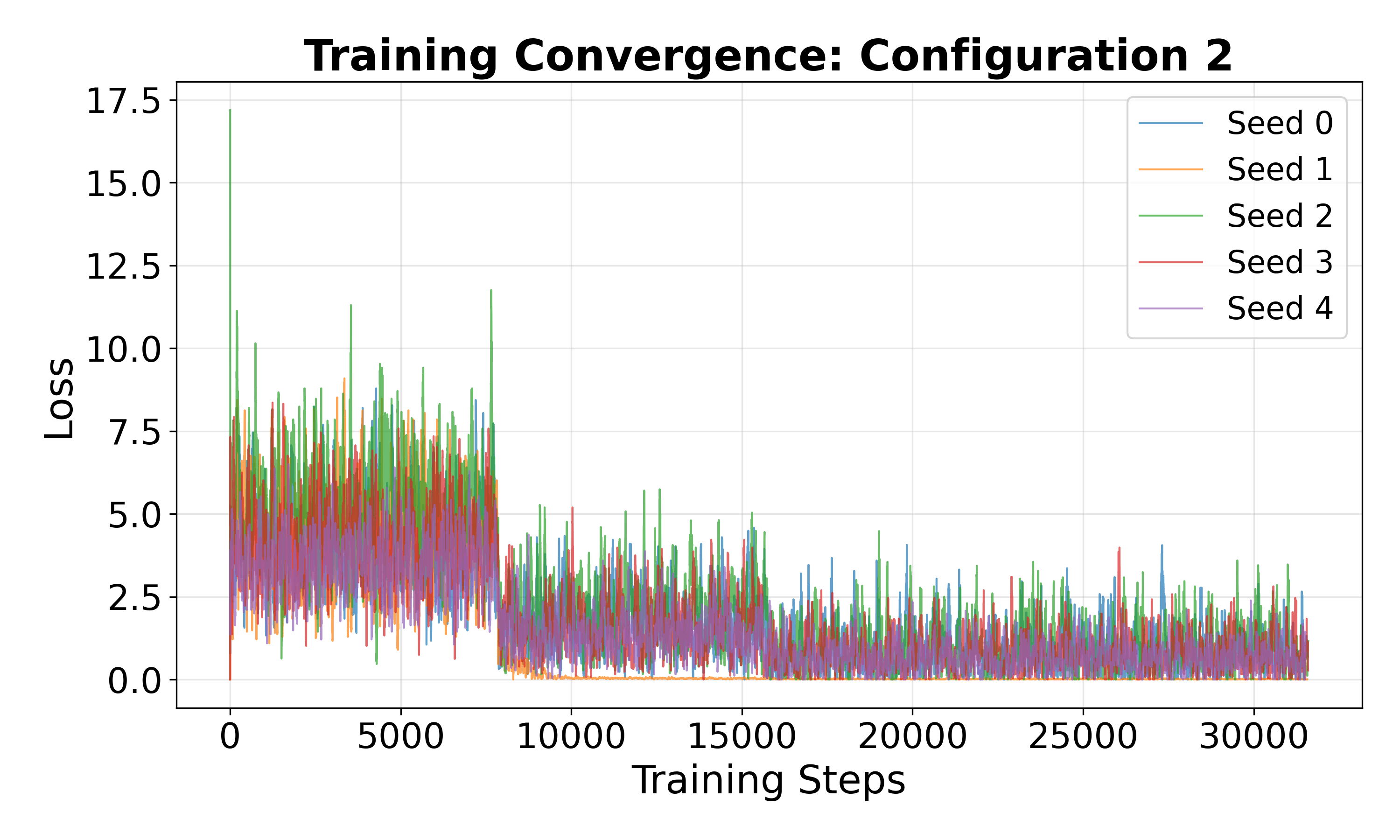}
    \caption*{Figure 3: Per-batch training loss across 5 seeds for Configuration 2.}
    \label{FIG:TRAIN_CONFIG2}
\end{figure}

Results also provided in Table~\hyperref[TABLE:RESULTS]{I}. Fig.~\hyperref[FIG:TRAIN_CONFIG1]{2} and Fig.~\hyperref[FIG:TRAIN_CONFIG2]{3} visualizes 
the reduction in training loss across different training steps (batches) for multiple random seeds in both configurations. The convergence demonstrates the learning dynamics, enabling the network to reliably recover the exact underlying automata structure from random initialization.

\section{Conclusion}
\label{sec:conclusion}

In this work, we have established a rigorous structural equivalence between neural networks and Alternating Finite Automata (AFAs), effectively bridging the gap between differentiable statistical learning and succinct symbolic reasoning. By introducing the Logic-Gated Time-Shared Feedforward Network, we demonstrated that a simple, learnable bias mechanism is sufficient to transform standard recurrent layers into differentiable boolean circuits capable of seamlessly switching between Existential (\textsc{OR}) and Universal (\textsc{AND}) aggregation.

Our theoretical analysis proves that this architecture inherits the exponential succinctness of AFAs, enabling a network of linear width $n$ to represent complex regular languages that would otherwise require hyper-exponential state spaces in classical deterministic or nondeterministic models. Crucially, we showed that this expressivity does not come at the cost of learnability. Through a continuous relaxation of the logic gates, our framework allows for the simultaneous discovery of automata topology and logical semantics directly from binary labels via standard gradient descent.

These findings challenge the view of neural networks as purely approximate pattern matchers, revealing their latent capacity for exact, hierarchical boolean logic. By grounding deep learning primitives in formal automata theory, we provide a pathway toward neuro-symbolic models that combine the efficiency of neural training with the interpretability and verifiability of formal logic. Future work will explore extending this paradigm to context-free languages through differentiable memory structures and applying these succinct, verifiable models to safety-critical domains where logical guarantees are paramount.

\section{Limitations}
\label{sec:limitations}

Our work establishes a rigorous and efficient framework for simulating and learning Alternating Finite Automata using Logic-Gated TS-FFNs. While we demonstrate robust learnability across diverse configurations, there remain exciting avenues for future research. For example, our gradient-based learning approach relies on a continuous relaxation of boolean logic. While our curriculum strategy effectively navigates this landscape to recover exact automata, future theoretical work could further characterize the optimization bounds of this relaxation, potentially leading to even faster convergence guarantees for larger-scale logical synthesis.

\section{Broader Impact}
\label{sec:broader_impact}
This work advances the theoretical understanding of how neural networks can implement symbolic computation, offering a rigorous bridge between automata theory and deep learning. The framework has potential implications for neural-symbolic reasoning, formal verification, program synthesis, natural language understanding, and interpretable AI systems. By enabling feedforward architectures to perform exact symbolic reasoning with transparent internal representations, this approach supports applications requiring correctness, verifiability, and alignment with formal logic. As a foundational theoretical contribution, the work does not pose foreseeable ethical or societal risks.

\subsection*{Acknowledgment: GenAI Use}
ChatGPT~\cite{chatgpt} has been used to improve the syntax and grammar of several sections in the manuscript. It has also been used to identify and review literature and relevant papers on deep learning in automata learning.

\bibliographystyle{plainnat}
\bibliography{neurips_2025}

@article{SIEGELMANN199177,
title = {Turing computability with neural nets},
journal = {Applied Mathematics Letters},
volume = {4},
number = {6},
pages = {77-80},
year = {1991},
issn = {0893-9659},
doi = {https://doi.org/10.1016/0893-9659(91)90080-F},
url = {https://www.sciencedirect.com/science/article/pii/089396599190080F},
author = {Hava T. Siegelmann and Eduardo D. Sontag}
}

@article{McCulloch1943,
  title={A logical calculus of the ideas immanent in nervous activity},
  author={McCulloch, Warren S. and Pitts, Walter},
  journal={The bulletin of mathematical biophysics},
  volume={5},
  number={4},
  pages={115--133},
  year={1943},
  publisher={Springer},
  doi={10.1007/BF02478259},
  url={https://doi.org/10.1007/BF02478259}
}

@article{OMLIN199641,
title = {Extraction of rules from discrete-time recurrent neural networks},
journal = {Neural Networks},
volume = {9},
number = {1},
pages = {41-52},
year = {1996},
issn = {0893-6080},
doi = {https://doi.org/10.1016/0893-6080(95)00086-0},
url = {https://www.sciencedirect.com/science/article/pii/0893608095000860},
author = {Christian W. Omlin and C.Lee Giles}
}

@ARTICLE{Giles1992,
  author={Giles, C. L. and Miller, C. B. and Chen, D. and Chen, H. H. and Sun, G. Z. and Lee, Y. C.},
  journal={Neural Computation}, 
  title={Learning and Extracting Finite State Automata with Second-Order Recurrent Neural Networks}, 
  year={1992},
  volume={4},
  number={3},
  pages={393-405},
  doi={10.1162/neco.1992.4.3.393}}

@incollection{TINO1998171,
title = {Chapter 6 - Finite State Machines and Recurrent Neural Networks — Automata and Dynamical Systems Approaches},
editor = {Omid Omidvar and Judith Dayhoff},
booktitle = {Neural Networks and Pattern Recognition},
publisher = {Academic Press},
address = {San Diego},
pages = {171-219},
year = {1998},
isbn = {978-0-12-526420-4},
doi = {https://doi.org/10.1016/B978-012526420-4/50007-0},
url = {https://www.sciencedirect.com/science/article/pii/B9780125264204500070},
author = {Peter Tiňo and Bill G. Horne and C. Lee Giles and Pete C. Collingwood}
}

@InProceedings{pmlr-v80-weiss18a,
  title = 	 {Extracting Automata from Recurrent Neural Networks Using Queries and Counterexamples},
  author =       {Weiss, Gail and Goldberg, Yoav and Yahav, Eran},
  booktitle = 	 {Proceedings of the 35th International Conference on Machine Learning},
  pages = 	 {5247--5256},
  year = 	 {2018},
  editor = 	 {Dy, Jennifer and Krause, Andreas},
  volume = 	 {80},
  series = 	 {Proceedings of Machine Learning Research},
  month = 	 {10-15 Jul},
  publisher =    {PMLR},
  url = 	 {https://proceedings.mlr.press/v80/weiss18a.html}
}

@book{hopcroft2001introduction,
author = {Hopcroft, John E. and Motwani, Rajeev and Ullman, Jeffrey D.},
title = {Introduction to Automata Theory, Languages, and Computation (3rd Edition)},
year = {2006},
isbn = {0321455363},
publisher = {Addison-Wesley Longman Publishing Co., Inc.},
address = {USA}
}

@inproceedings{bhattamishra-etal-2020-ability,
    title = "On the {A}bility and {L}imitations of {T}ransformers to {R}ecognize {F}ormal {L}anguages",
    author = "Bhattamishra, Satwik  and
      Ahuja, Kabir  and
      Goyal, Navin",
    editor = "Webber, Bonnie  and
      Cohn, Trevor  and
      He, Yulan  and
      Liu, Yang",
    booktitle = "Proceedings of the 2020 Conference on Empirical Methods in Natural Language Processing (EMNLP)",
    month = nov,
    year = "2020",
    address = "Online",
    publisher = "Association for Computational Linguistics",
    url = "https://aclanthology.org/2020.emnlp-main.576/",
    doi = "10.18653/v1/2020.emnlp-main.576",
    pages = "7096--7116"
}

@article{Hahn_2020,
   title={Theoretical Limitations of Self-Attention in Neural Sequence Models},
   volume={8},
   ISSN={2307-387X},
   url={http://dx.doi.org/10.1162/tacl_a_00306},
   DOI={10.1162/tacl_a_00306},
   journal={Transactions of the Association for Computational Linguistics},
   publisher={MIT Press - Journals},
   author={Hahn, Michael},
   year={2020},
   month=dec, pages={156–171} }

@inproceedings{merrill-etal-2020-formal,
    title = "A Formal Hierarchy of {RNN} Architectures",
    author = "Merrill, William  and
      Weiss, Gail  and
      Goldberg, Yoav  and
      Schwartz, Roy  and
      Smith, Noah A.  and
      Yahav, Eran",
    editor = "Jurafsky, Dan  and
      Chai, Joyce  and
      Schluter, Natalie  and
      Tetreault, Joel",
    booktitle = "Proceedings of the 58th Annual Meeting of the Association for Computational Linguistics",
    month = jul,
    year = "2020",
    address = "Online",
    publisher = "Association for Computational Linguistics",
    url = "https://aclanthology.org/2020.acl-main.43/",
    doi = "10.18653/v1/2020.acl-main.43",
    pages = "443--459"
}

@inproceedings{merrill-2019-sequential,
    title = "Sequential Neural Networks as Automata",
    author = "Merrill, William",
    editor = "Eisner, Jason  and
      Gall{\'e}, Matthias  and
      Heinz, Jeffrey  and
      Quattoni, Ariadna  and
      Rabusseau, Guillaume",
    booktitle = "Proceedings of the Workshop on Deep Learning and Formal Languages: Building Bridges",
    month = aug,
    year = "2019",
    address = "Florence",
    publisher = "Association for Computational Linguistics",
    url = "https://aclanthology.org/W19-3901/",
    doi = "10.18653/v1/W19-3901",
    pages = "1--13"
}

@inproceedings{Gu2021,
author = {Gu, Albert and Johnson, Isys and Goel, Karan and Saab, Khaled and Dao, Tri and Rudra, Atri and R\'{e}, Christopher},
title = {Combining recurrent, convolutional, and continuous-time models with linear state-space layers},
year = {2021},
isbn = {9781713845393},
publisher = {Curran Associates Inc.},
address = {Red Hook, NY, USA},
booktitle = {Proceedings of the 35th International Conference on Neural Information Processing Systems},
articleno = {44},
numpages = {14},
series = {NIPS '21}
}

@inproceedings{Dao2022,
author = {Dao, Tri and Fu, Daniel Y. and Ermon, Stefano and Rudra, Atri and R\'{e}, Christopher},
title = {FLASHATTENTION: fast and memory-efficient exact attention with IO-awareness},
year = {2022},
isbn = {9781713871088},
publisher = {Curran Associates Inc.},
address = {Red Hook, NY, USA},
booktitle = {Proceedings of the 36th International Conference on Neural Information Processing Systems},
articleno = {1189},
numpages = {16},
location = {New Orleans, LA, USA},
series = {NIPS '22}
}

@ARTICLE{Rabin1959,
  author={Rabin, M. O. and Scott, D.},
  journal={IBM Journal of Research and Development}, 
  title={Finite Automata and Their Decision Problems}, 
  year={1959},
  volume={3},
  number={2},
  pages={114-125},
  doi={10.1147/rd.32.0114}}

@ARTICLE{Casey1996,
  author={Casey, Mike},
  journal={Neural Computation}, 
  title={The Dynamics of Discrete-Time Computation, with Application to Recurrent Neural Networks and Finite State Machine Extraction}, 
  year={1996},
  volume={8},
  number={6},
  pages={1135-1178},
  doi={10.1162/neco.1996.8.6.1135}}

@INPROCEEDINGS{243123,
  author={Guo-Zheng Sun},
  booktitle={IEE Colloquium on Grammatical Inference: Theory, Applications and Alternatives}, 
  title={Learning context-free grammar with enhanced neural network pushdown automaton}, 
  year={1993},
  volume={},
  number={},
  pages={P6/1-P613},
  doi={}}

@ARTICLE{dhayalkar2025nfa,
  author={Dhayalkar, Sahil Rajesh},
  journal={IEEE Access}, 
  title={A Constructive Framework for Nondeterministic Automata via Time-Shared, Depth-Unrolled Feedforward Networks}, 
  year={2026},
  volume={14},
  number={},
  pages={8903-8917},
  doi={10.1109/ACCESS.2026.3654057}}

@misc{dhayalkar2025neuralnetworksuniversalfinitestate,
      title={Neural Networks as Universal Finite-State Machines: A Constructive Deterministic Finite Automaton Theory}, 
      author={Sahil Rajesh Dhayalkar},
      year={2025},
      eprint={2505.11694},
      archivePrefix={arXiv},
      primaryClass={cs.LG},
      url={https://arxiv.org/abs/2505.11694}, 
}

@misc{dhayalkar2025symbolicfeedforwardnetworksprobabilistic,
      title={Symbolic Feedforward Networks for Probabilistic Finite Automata: Exact Simulation and Learnability}, 
      author={Sahil Rajesh Dhayalkar},
      year={2025},
      eprint={2509.10034},
      archivePrefix={arXiv},
      primaryClass={cs.LG},
      url={https://arxiv.org/abs/2509.10034}, 
}

@article{Chandra1981,
author = {Chandra, Ashok K. and Kozen, Dexter C. and Stockmeyer, Larry J.},
title = {Alternation},
year = {1981},
issue_date = {Jan. 1981},
publisher = {Association for Computing Machinery},
address = {New York, NY, USA},
volume = {28},
number = {1},
issn = {0004-5411},
url = {https://doi.org/10.1145/322234.322243},
doi = {10.1145/322234.322243},
journal = {J. ACM},
month = jan,
pages = {114–133},
numpages = {20}
}

@inproceedings{Kozen1976,
author = {Kozen, Dexter},
title = {On parallelism in turing machines},
year = {1976},
publisher = {IEEE Computer Society},
address = {USA},
url = {https://doi.org/10.1109/SFCS.1976.20},
doi = {10.1109/SFCS.1976.20},
booktitle = {Proceedings of the 17th Annual Symposium on Foundations of Computer Science},
pages = {89–97},
numpages = {9},
series = {SFCS '76}
}

@article{Fellah01011990,
author = {A. Fellah and H. Jürgensen and S. Yu},
title = {Constructions for alternating finite automata},
journal = {International Journal of Computer Mathematics},
volume = {35},
number = {1-4},
pages = {117--132},
year = {1990},
publisher = {Taylor \& Francis},
doi = {10.1080/00207169008803893},
URL = { 
        https://doi.org/10.1080/00207169008803893
},
eprint = { 
        https://doi.org/10.1080/00207169008803893
}}

@article{Evans2018,
author = {Evans, Richard and Grefenstette, Edward},
title = {Learning explanatory rules from noisy data},
year = {2018},
issue_date = {January 2018},
publisher = {AI Access Foundation},
address = {El Segundo, CA, USA},
volume = {61},
number = {1},
issn = {1076-9757},
journal = {J. Artif. Int. Res.},
month = jan,
pages = {1–64},
numpages = {64}
}

@misc{dong2019neurallogicmachines,
      title={Neural Logic Machines}, 
      author={Honghua Dong and Jiayuan Mao and Tian Lin and Chong Wang and Lihong Li and Denny Zhou},
      year={2019},
      eprint={1904.11694},
      archivePrefix={arXiv},
      primaryClass={cs.AI},
      url={https://arxiv.org/abs/1904.11694}, 
}

@misc{serafini2016logictensornetworksdeep,
      title={Logic Tensor Networks: Deep Learning and Logical Reasoning from Data and Knowledge}, 
      author={Luciano Serafini and Artur d'Avila Garcez},
      year={2016},
      eprint={1606.04422},
      archivePrefix={arXiv},
      primaryClass={cs.AI},
      url={https://arxiv.org/abs/1606.04422}, 
}

@inproceedings{Donadello2017,
author = {Donadello, Ivan and Serafini, Luciano and Garcez, Artur D'Avila},
title = {Logic tensor networks for semantic image interpretation},
year = {2017},
isbn = {9780999241103},
publisher = {AAAI Press},
booktitle = {Proceedings of the 26th International Joint Conference on Artificial Intelligence},
pages = {1596–1602},
numpages = {7},
location = {Melbourne, Australia},
series = {IJCAI'17}
}

@misc{bengio2013estimatingpropagatinggradientsstochastic,
      title={Estimating or Propagating Gradients Through Stochastic Neurons for Conditional Computation}, 
      author={Yoshua Bengio and Nicholas Léonard and Aaron Courville},
      year={2013},
      eprint={1308.3432},
      archivePrefix={arXiv},
      primaryClass={cs.LG},
      url={https://arxiv.org/abs/1308.3432}, 
}

@misc{petersen2021differentiablesortingnetworksscalable,
      title={Differentiable Sorting Networks for Scalable Sorting and Ranking Supervision}, 
      author={Felix Petersen and Christian Borgelt and Hilde Kuehne and Oliver Deussen},
      year={2021},
      eprint={2105.04019},
      archivePrefix={arXiv},
      primaryClass={cs.LG},
      url={https://arxiv.org/abs/2105.04019}, 
}

@inproceedings{
liu2018darts,
title={{DARTS}: Differentiable Architecture Search},
author={Hanxiao Liu and Karen Simonyan and Yiming Yang},
booktitle={International Conference on Learning Representations},
year={2019},
url={https://openreview.net/forum?id=S1eYHoC5FX},
}

@inproceedings{
xie2018snas,
title={{SNAS}: stochastic neural architecture search},
author={Sirui Xie and Hehui Zheng and Chunxiao Liu and Liang Lin},
booktitle={International Conference on Learning Representations},
year={2019},
url={https://openreview.net/forum?id=rylqooRqK7},
}

@inproceedings{Vardi1996,
author = {Vardi, Moshe Y.},
title = {An automata-theoretic approach to linear temporal logic},
year = {1996},
isbn = {3540609156},
publisher = {Springer-Verlag},
address = {Berlin, Heidelberg},
booktitle = {Proceedings of the VIII Banff Higher Order Workshop Conference on Logics for Concurrency: Structure versus Automata: Structure versus Automata},
pages = {238–266},
numpages = {29},
location = {Banff, Canada}
}

@article{Muller1987,
author = {Muller, David, E. and Schupp, Paul, E.},
title = {Alternating automata on infinite trees},
year = {1987},
issue_date = {Oct. 1987},
publisher = {Elsevier Science Publishers Ltd.},
address = {GBR},
volume = {54},
number = {2–3},
issn = {0304-3975},
url = {https://doi.org/10.1016/0304-3975(87)90133-2},
doi = {10.1016/0304-3975(87)90133-2},
journal = {Theor. Comput. Sci.},
month = oct,
pages = {267–276},
numpages = {10}
}

@article{Kupferman2000,
author = {Kupferman, Orna and Vardi, Moshe Y.},
title = {Weak alternating automata are not that weak},
year = {2001},
issue_date = {July 2001},
publisher = {Association for Computing Machinery},
address = {New York, NY, USA},
volume = {2},
number = {3},
issn = {1529-3785},
url = {https://doi.org/10.1145/377978.377993},
doi = {10.1145/377978.377993},
journal = {ACM Trans. Comput. Logic},
month = jul,
pages = {408–429},
numpages = {22}
}

@inproceedings{Wang2018,
author = {Wang, Shiqi and Pei, Kexin and Whitehouse, Justin and Yang, Junfeng and Jana, Suman},
title = {Efficient formal safety analysis of neural networks},
year = {2018},
publisher = {Curran Associates Inc.},
address = {Red Hook, NY, USA},
booktitle = {Proceedings of the 32nd International Conference on Neural Information Processing Systems},
pages = {6369–6379},
numpages = {11},
location = {Montr\'{e}al, Canada},
series = {NIPS'18}
}

@inproceedings{
selsam2018learning,
title={Learning a {SAT} Solver from Single-Bit Supervision},
author={Daniel Selsam and Matthew Lamm and Benedikt B\"{u}nz and Percy Liang and Leonardo de Moura and David L. Dill},
booktitle={International Conference on Learning Representations},
year={2019},
url={https://openreview.net/forum?id=HJMC_iA5tm},
}

@misc{
amizadeh2020pdp,
title={{\{}PDP{\}}: A General Neural Framework for Learning {\{}SAT{\}} Solvers},
author={Saeed Amizadeh and Sergiy Matusevych and Markus Weimer},
year={2020},
url={https://openreview.net/forum?id=S1xaf6VFPB}
}

@INPROCEEDINGS{7410480,
  author={He, Kaiming and Zhang, Xiangyu and Ren, Shaoqing and Sun, Jian},
  booktitle={2015 IEEE International Conference on Computer Vision (ICCV)}, 
  title={Delving Deep into Rectifiers: Surpassing Human-Level Performance on ImageNet Classification}, 
  year={2015},
  volume={},
  number={},
  pages={1026-1034},
  doi={10.1109/ICCV.2015.123}}

@misc{kingma2017adammethodstochasticoptimization,
      title={Adam: A Method for Stochastic Optimization}, 
      author={Diederik P. Kingma and Jimmy Ba},
      year={2017},
      eprint={1412.6980},
      archivePrefix={arXiv},
      primaryClass={cs.LG},
      url={https://arxiv.org/abs/1412.6980}, 
}

@inproceedings{pytorch,
  author    = {Paszke, Adam and Gross, Sam and Massa, Francisco and Lerer, Adam and Bradbury, James and Chanan, Gregory and Killeen, Trevor and Lin, Zeming and Gimelshein, Natalia and Antiga, Luca and Desmaison, Alban and Kopf, Andreas and Yang, Edward and DeVito, Zachary and Raison, Martin and Tejani, Alykhan and Chilamkurthy, Sasank and Steiner, Benoit and Fang, Lu and Bai, Junjie and Chintala, Soumith},
  title     = {PyTorch: An Imperative Style, High-Performance Deep Learning Library},
  booktitle = {Advances in Neural Information Processing Systems},
  editor    = {Wallach, H. and Larochelle, H. and Beygelzimer, A. and d'Alch{\'{e}}{-}Buc, F. and Fox, E. and Garnett, R.},
  volume    = {32},
  pages     = {8024--8035},
  year      = {2019},
  url       = {https://proceedings.neurips.cc/paper_files/paper/2019/file/bdbca288fee7f92f2bfa9f7012727740-Paper.pdf}
}

@misc{chatgpt,
  author = {OpenAI},
  title = {ChatGPT},
  year = {2023},
  month = {March},
  note = {Large language model},
  url = {https://chatgpt.com/}
}

\appendix
\section{Proofs of Theoretical Results}
\label{app:proofs}

\subsection{Proof Sketch of Proposition~\ref{PROP:LOGICAL_AGGREGATION}: Logical Aggregation via Biased Linear Units}
\label{PROOF:PROP:LOGICAL_AGGREGATION}

\begin{proof}
Let $q_i \in Q$ be a state and $x \in \Sigma \cup \{\varepsilon\}$ be a transition symbol. Let $P_i$ be the set of predecessors (incoming states) for $q_i$ on symbol $x$, and let $d_i = |P_i|$ be the in-degree. Let $v_{prev} \in \{0,1\}^n$ be the boolean activation vector of the previous step.

We aim to show that the linear threshold operation:
\begin{equation}
    v_{curr}[i] = \sigma \left( \sum_{j=1}^n T_{ij}^{(x)} [v_{prev}]_j - \beta_i^{(x)} \right)
\end{equation}
correctly aggregates the incoming boolean signals according to the specified logic, where $\sigma(z) = \mathbf{1}_{[z \ge 0]}$.

First, consider the linear term inside the activation function. By definition of the forward transition matrix $T^{(x)}$, $T_{ij}^{(x)} = 1$ if and only if there is a transition $q_j \xrightarrow{x} q_i$, and 0 otherwise. Thus, the summation term simplifies to a count of active predecessors:
\begin{equation}
    Z_i = \sum_{j=1}^n T_{ij}^{(x)} [v_{prev}]_j = \sum_{q_j \in P_i} [v_{prev}]_j = | \{ q_j \in P_i \mid q_j \text{ is active} \} |
\end{equation}
This integer $Z_i$ represents the number of incoming paths that are active. We analyze the two logical cases:

\paragraph{Case 1: Existential Aggregation (\textsc{OR}):}
To implement logical \textsc{OR} (Existential reachability), the state $q_i$ should activate if at least one predecessor is active. The condition is:
\begin{equation}
    Active(q_i) \iff Z_i \ge 1
\end{equation}
We set the bias $\beta_i^{(x)} = 0.5$. The neuron's activation condition becomes:
\begin{equation}
    Z_i - 0.5 \ge 0 \implies Z_i \ge 0.5
\end{equation}
Since $Z_i$ is an integer, $Z_i \ge 0.5$ is strictly equivalent to $Z_i \ge 1$. Thus, the neuron fires if any incoming path is active.

\paragraph{Case 2: Universal Aggregation (\textsc{AND}):}
To implement logical \textsc{AND} (Universal synchronization), the state $q_i$ should activate only if all predecessors are active. The condition is:
\begin{equation}
    Active(q_i) \iff Z_i = d_i
\end{equation}
We set the bias $\beta_i^{(x)} = d_i - 0.5$. The neuron's activation condition becomes:
\begin{equation}
    Z_i - (d_i - 0.5) \ge 0 \implies Z_i \ge d_i - 0.5
\end{equation}
Since $Z_i$ is bounded by $d_i$ (it cannot exceed the number of incoming edges), the condition $Z_i \ge d_i - 0.5$ can only be satisfied if $Z_i = d_i$. Thus, the neuron fires only if all incoming paths are active.

\paragraph{Conclusion:}
In both cases, the linear threshold operation with the connectivity-dependent bias $\beta_i^{(x)}$ correctly implements the desired boolean aggregation logic (\textsc{OR} vs. \textsc{AND}) over the input signals.
\end{proof}

\subsection{Proof Sketch of Theorem~\ref{THM:AFA_SIMULATION}: Simulation of AFAs by Logic-Gated TS-FFNs}
\label{PROOF:THM:AFA_SIMULATION}

\begin{proof}
Let $\mathcal{A} = (Q, \Sigma \cup \{\varepsilon\}, \delta, q_0, g, F)$ be an Alternating Finite Automaton with $n$ states. Let $x = x_1 x_2 \dots x_L$ be an input string. We construct a Logic-Gated TS-FFN $f_\theta$ with parameters defined as in Proposition \ref{PROP:LOGICAL_AGGREGATION}. We identify the boolean activation status of the states with a vector $v \in \{0,1\}^n$.

We prove by induction on the time step $t$ (from $0$ up to $L$) that the network state $v_t$ correctly encodes the set of active states after processing the prefix $x_1 \dots x_t$.

\paragraph{Base Case ($t=0$):}
At step $t=0$, no input symbols have been consumed. The active states are the start state $q_0$ and any states reachable from $q_0$ via $\varepsilon$-transitions.
The network initializes $v_0 = \mathcal{C}_\varepsilon(e_{q_0})$.
\begin{itemize}
    \item The vector $e_{q_0}$ correctly marks the start state.
    \item The closure operator $\mathcal{C}_\varepsilon$ iteratively applies the update $u \leftarrow \sigma(T^{(\varepsilon)} u - \beta^{(\varepsilon)})$ as depicted in~\cite{dhayalkar2025nfa}. As shown in Proposition \ref{PROP:LOGICAL_AGGREGATION}, each iteration propagates activation forward across one layer of $\varepsilon$-transitions, aggregating signals according to the logical type ($g(q)$) of each state (e.g., waiting for all incoming $\varepsilon$-paths if configured as \textsc{AND}).
    \item Since the state space is finite, this fixed-point iteration converges to the exact set of initially active states.
\end{itemize}
Thus, $v_0$ correctly encodes the initial configuration.

\paragraph{Inductive Step:}
Assume that at step $t-1$, the vector $v_{t-1}$ correctly encodes the active states after processing the prefix $x_1 \dots x_{t-1}$. We must show that $v_t$ correctly encodes the active states after processing $x_t$.
The network update is given by Equation (5):
\begin{equation}
    v_t = \mathcal{C}_\varepsilon \left( \sigma\left(T^{(x_t)} v_{t-1} - \beta^{(x_t)}\right) \right)
\end{equation}
Let $v'_{t} = \sigma(T^{(x_t)} v_{t-1} - \beta^{(x_t)})$ be the intermediate value before $\varepsilon$-closure.
\begin{enumerate}
    \item Symbol Transition: By Proposition \ref{PROP:LOGICAL_AGGREGATION}, the operation $\sigma(T^{(x_t)} v_{t-1} - \beta^{(x_t)})$ correctly computes the boolean activation of each state $q_i$ based on its predecessors on symbol $x_t$.
    \begin{itemize}
        \item If configured as \textsc{OR} ($\beta \approx 0.5$), $q_i$ activates if any predecessor is active.
        \item If configured as \textsc{AND} ($\beta \approx d_{in}-0.5$), $q_i$ activates only if all predecessors are active.
    \end{itemize}
    Thus, $v'_{t}$ encodes the states activated by the transition on $x_t$.

    \item Epsilon Closure: The subsequent application of $\mathcal{C}_\varepsilon$~\cite{dhayalkar2025nfa} accounts for states that become active via $\varepsilon$-transitions immediately following the consumption of $x_t$. The logic is identical to the base case: the closure propagates the activation of $v'_{t}$ downstream across $\varepsilon$-edges according to the \textsc{AND}/\textsc{OR} aggregation semantics.
\end{enumerate}
Therefore, $v_t$ correctly combines symbol consumption and $\varepsilon$-transitions, accurately reflecting the recursive definition of state activation.

\paragraph{Conclusion:}
By induction, $v_L$ correctly encodes the active states after processing the full string $x_1 \dots x_L$. The network output $f_\theta(x) = \mathbf{1}_{[v_L \cdot \mathbf{1}_F^T > 0]}$ checks if any final state is active. Thus, the network simulates the forward dynamics of the automaton.
\end{proof}

\subsection{Proof Sketch of Proposition~\ref{PROP:SUCCINCTNESS}: Succinctness and Parameter Efficiency}
\label{PROOF:PROP:SUCCINCTNESS}

\begin{proof}
Let $\mathcal{A} = (Q, \Sigma \cup \{\varepsilon\}, \delta, q_0, g, F)$ be an Alternating Finite Automaton with $|Q|=n$ states and an input alphabet size $|\Sigma|=k$. We construct a Logic-Gated TS-FFN $f_\theta$ that simulates $\mathcal{A}$ according to Theorem \ref{THM:AFA_SIMULATION}. We analyze the two claims separately.

\paragraph{1. Parameter Efficiency:}
The Logic-Gated TS-FFN is parameterized solely by the time-shared matrices and bias vectors used in the layer-wise update rule.
\begin{itemize}
    \item Transition Matrices: For each symbol $x \in \Sigma$, there is a transition matrix $T^{(x)} \in \mathbb{R}^{n \times n}$. Additionally, there is one shared $\varepsilon$-transition matrix $T^{(\varepsilon)} \in \mathbb{R}^{n \times n}$. The total number of transition matrices is $k+1$.
    \item Bias Vectors: For each symbol $x \in \Sigma$, there is a bias vector $\beta^{(x)} \in \mathbb{R}^n$. Additionally, there is one shared $\varepsilon$-bias vector $\beta^{(\varepsilon)} \in \mathbb{R}^n$. The total number of bias vectors is $k+1$.
    \item Initialization and Readout: The initialization requires a vector $e_{q_0} \in \{0,1\}^n$ (fixed or learnable) and the readout requires a mask for $F$ (fixed or learnable). These contribute $\mathcal{O}(n)$ parameters.
\end{itemize}
The total parameter count $N_{params}$ is given by:
\begin{equation}
    N_{params} = \underbrace{(k+1) \cdot n^2}_{\text{Transition Matrices}} + \underbrace{(k+1) \cdot n}_{\text{Bias Vectors}} + \underbrace{\mathcal{O}(n)}_{\text{Boundary}}
\end{equation}
Simplifying, $N_{params} = \mathcal{O}((k+1)(n^2 + n)) = \mathcal{O}(k n^2)$.
Crucially, these parameters are shared across all time steps $t=1, \dots, L$ of the unrolled computation graph. Therefore, the parameter count is strictly independent of the input string length $L$.

\paragraph{2. Succinctness:}
Succinctness is measured by comparing the network width (number of neurons per layer) required to recognize a language $\mathcal{L}$ against standard baselines.
\begin{itemize}
    \item Standard Baseline (NFA Simulation): A standard feedforward network (like a vanilla TS-FFN or RNN) simulating a regular language typically corresponds to an NFA simulation. It is a well-known result in automata theory that converting an $n$-state AFA to an equivalent NFA may require $2^n$ states (via the subset construction). To simulate this NFA directly, a standard network would require a width of $2^n$ neurons to represent the $2^n$ distinct NFA states using a one-hot or distributed reachability encoding.
    
    \item Logic-Gated TS-FFN: Our architecture uses a width of exactly $n$ neurons. The state of the network at any depth $t$ is a vector $v_t \in \{0,1\}^n$. This vector represents the boolean activation configuration of the $n$ AFA states.
    \item Mechanism: By utilizing the forward-unrolling mechanism and the logic-gated biases $\beta$, the network computes the boolean function $f: \{0,1\}^n \to \{0,1\}^n$ required to update this configuration directly. This allows the network to perform the equivalent of the NFA's powerset transition (tracking the boolean combination of states) within a layer of width $n$.
\end{itemize}
Thus, the Logic-Gated TS-FFN achieves an exponential reduction in spatial complexity (network width) relative to a standard NFA-simulating network for the same class of languages.
\end{proof}

\subsection{Proof Sketch of Theorem~\ref{THM:EQUIVALENCE}: Equivalence between Logic-Gated TS-FFNs and AFAs}
\label{PROOF:THM:EQUIVALENCE}

\begin{proof}
The equivalence holds constructively in both directions.

\paragraph{Forward Direction:} Given a regular language $\mathcal{L}$, let $\mathcal{A} = (Q, \Sigma \cup \{\varepsilon\}, \delta, q_0, g, F)$ be an AFA recognizing $\mathcal{L}$ with $|Q|=n$. Construct an LG-TS-FFN $f_{\theta}$ as follows:
\begin{itemize}
    \item Represent the boolean activation configuration of states at time $t$ as a vector $s_t$, initialized with the one-hot start vector $e_{q_0}$.
    \item Compute the initial $\varepsilon$-closure: $s_0 = \mathcal{C}_\varepsilon(e_{q_0})$.
    \item For each symbol $x_t$ in the input string, compute the update:
    \begin{equation}
        s_t = \mathcal{C}_\varepsilon \left( \sigma\left(T^{(x_t)} s_{t-1} - \beta^{(x_t)}\right) \right)
    \end{equation}
    where parameters are defined explicitly by the AFA structure:
    \begin{itemize}
        \item $T_{ij}^{(x)} = 1 \iff q_i \in \delta(q_j, x)$.
        \item $\beta_i^{(x)} = 0.5$ if $g(q_i) = \vee$ (Existential), and $\beta_i^{(x)} = d_{in}-0.5$ if $g(q_i) = \wedge$ (Universal).
    \end{itemize}
    \item Accept if $\langle s_L, \mathbf{1}_F \rangle > 0$, where $\mathbf{1}_F$ indicates the accepting states.
\end{itemize}
Because every operation corresponds exactly to the AFA's transition function and logic gates, the network accepts precisely those strings in $\mathcal{L}$.

\paragraph{Reverse Direction:} Given an LG-TS-FFN $f_{\theta}$ constructed by this symbolic procedure:
\begin{itemize}
    \item Each matrix $T^{(x)} \in \{0,1\}^{n \times n}$ defines a transition relation: $q_i \in \delta(q_j, x)$ iff $T_{ij}^{(x)}=1$.
    \item Each bias vector $\beta^{(x)}$ defines the logical types $g$:
    \begin{equation}
        g(q_i) = \begin{cases} \vee & \text{if } \beta_i^{(x)} \approx 0.5 \\ \wedge & \text{if } \beta_i^{(x)} \approx d_{in} - 0.5 \end{cases}
    \end{equation}
    \item The initial state corresponds to the vector $e_{q_0}$, and final states are encoded by $\mathbf{1}_F$.
\end{itemize}
Therefore, the LG-TS-FFN precisely simulates an AFA $\mathcal{A}'$ whose accepted language satisfies $f_{\theta}(x) = 1 \iff x \in \mathcal{L}(\mathcal{A}')$.

\paragraph{Conclusion:} The equivalence between LG-TS-FFNs and AFAs holds constructively in both directions: any AFA can be simulated exactly by an LG-TS-FFN, and any such LG-TS-FFN corresponds to an AFA recognizing the same language.
\end{proof}

\subsection{Proof Sketch of Proposition~\ref{PROP:LEARNABILITY}: Gradient-Based Learnability of Logic Gates}
\label{PROOF:PROP:LEARNABILITY}
\begin{proof}
Let $\mathcal{A} = (Q, \Sigma \cup \{\varepsilon\}, \delta, q_0, g, F)$ be the target Alternating Finite Automaton with $n$ states. Let $\mathcal{D} = \{(x^{(i)}, y^{(i)})\}_{i=1}^m$ be a dataset of input strings $x^{(i)} \in \Sigma^*$ and binary labels $y^{(i)} \in \{0,1\}$ indicating acceptance by $\mathcal{A}$.

We consider a relaxed Logic-Gated TS-FFN $f_\theta$ parameterized by $\theta = \{T^{(x)}, \beta^{(x)}\}_{x \in \Sigma \cup \{\varepsilon\}}$. The relaxation replaces the binary step activation with the logistic sigmoid function $\sigma_\lambda(z) = (1 + e^{-\lambda z})^{-1}$ and allows parameters to take real values.

We prove learnability in two parts: (1) Realizability (existence of a solution) and (2) Optimization Dynamics.

\paragraph{1. Constructive Realizability (Existence of $\theta^*$):}
By Theorem \ref{THM:AFA_SIMULATION}, there exists a discrete parameter configuration that exactly simulates $\mathcal{A}$ in the forward direction. We construct a corresponding continuous parameter set $\theta^*$ for the relaxed network:
\begin{itemize}
    \item Transitions: For every $x \in \Sigma \cup \{\varepsilon\}$, set $T^{*(x)}_{ij} = \gamma$ if $q_i \in \delta(q_j, x)$ (encoding a transition $q_j \to q_i$) and $-\gamma$ otherwise, for a sufficiently large scaling factor $\gamma > 0$.
    \item Biases: Set $\beta^{*(x)}_i$ such that the decision boundary respects the logic gate type of state $q_i$:
    \begin{itemize}
        \item If $g(q_i) = \vee$ (Existential/\textsc{OR}), set $\beta^{*(x)}_i = 0.5 \gamma$. The condition $\sum_j T_{ij} v_j - \beta > 0$ becomes $\gamma (\sum v_j - 0.5) > 0$, equivalent to $\sum v_j \ge 1$ (at least one active predecessor).
        \item If $g(q_i) = \wedge$ (Universal/\textsc{AND}), set $\beta^{*(x)}_i = (d_{in}(q_i) - 0.5) \gamma$, where $d_{in}$ is the number of incoming transitions on $x$. The condition becomes $\gamma (\sum v_j - (d_{in} - 0.5)) > 0$, equivalent to $\sum v_j = d_{in}$ (all predecessors active).
    \end{itemize}
\end{itemize}
As the temperature parameter $\lambda \to \infty$ (or $\gamma \to \infty$), the sigmoid function $\sigma_\lambda$ approaches the step function. Thus, for sufficiently large $\gamma$, the relaxed network $f_{\theta^*}(x)$ outputs values arbitrarily close to the discrete truth values $y^{(i)}$. Consequently, the empirical risk $\mathcal{L}(\theta^*) \approx 0$ is a global minimum. This proves the problem is realizable within the hypothesis space.

\paragraph{2. Gradient-Based Optimization Dynamics:}
We define the learning objective as minimizing the binary cross-entropy loss:
\begin{equation}
    l(\theta) = - \frac{1}{m} \sum_{i=1}^m \left[ y^{(i)} \log f_\theta(x^{(i)}) + (1 - y^{(i)}) \log (1 - f_\theta(x^{(i)})) \right]
\end{equation}
The gradients with respect to the parameters are computed via backpropagation. Consider the gradient for a bias parameter $\beta_i^{(x)}$ at a specific layer $t$. By the chain rule:
\begin{equation}
    \frac{\partial l}{\partial \beta_i^{(x)}} = \frac{\partial l}{\partial v_{t}[i]} \cdot \frac{\partial v_{t}[i]}{\partial z_{t}[i]} \cdot (-1)
\end{equation}
where $z_t$ is the pre-activation sum. The term $\frac{\partial l}{\partial v_{t}[i]}$ represents the error signal propagated backwards from the output at step $L$.
\begin{itemize}
    \item If the network falsely accepts a negative example ($y^{(i)}=0, f_\theta \approx 1$), the gradient descent update will increase $\beta_i^{(x)}$ (to suppress activation). This pushes the logic gate towards universal (\textsc{AND}) behavior, effectively raising the threshold.
    \item If the network falsely rejects a positive example ($y^{(i)}=1, f_\theta \approx 0$), the update will decrease $\beta_i^{(x)}$ (to encourage activation). This pushes the logic gate towards existential (\textsc{OR}) behavior, lowering the threshold.
\end{itemize}
Furthermore, the parameters are time-shared. The total gradient for a parameter is the sum over all time steps $t$: $\nabla_{\theta} l = \sum_{t=1}^L \nabla_{\theta}^{(t)} l$. This aggregation ensures that evidence from the entire sequence length accumulates constructively, mitigating the vanishing gradient problem for the symbolic structure.

\paragraph{Conclusion:}
Since an exact solution $\theta^*$ exists, and the gradient descent updates move parameters in directions consistent with correcting logical errors (adjusting thresholds for \textsc{AND}/\textsc{OR} behavior and weights for connectivity), the TS-FFN is empirically learnable from labeled data.
\end{proof}

\subsection{Example: Encoding input strings within the Logic-Gated TS-FFN}
\label{EXAMPLE}

Consider an input string $x = aba$ and an automaton over $\Sigma = \{a, b\}$. The unrolled Logic-Gated TS-FFN computes:
\begin{align*}
    v_0 &= \mathcal{C}_\varepsilon(e_{q_0}), \\
    v_1 &= \mathcal{C}_\varepsilon \left( \sigma \left( T^{(a)} v_0 - \beta^{(a)} \right) \right), \\
    v_2 &= \mathcal{C}_\varepsilon \left( \sigma \left( T^{(b)} v_1 - \beta^{(b)} \right) \right), \\
    v_3 &= \mathcal{C}_\varepsilon \left( \sigma \left( T^{(a)} v_2 - \beta^{(a)} \right) \right).
\end{align*}
Finally, acceptance is determined via a dot product with the indicator vector of accepting states:
\begin{equation*}
    \text{Accept}(x) = \text{True} \iff \langle v_3, \mathbf{1}_F \rangle > 0.
\end{equation*}
This explicit separation of control flow (symbol sequence) from computation is fundamental to the symbolic fidelity and modularity of the TS-FFN framework. It allows the model to retain both exact automata semantics and differentiable compositional structure.
\end{document}